\documentclass[a4paper,reqno,12pt,english]{amsart} %reqno permite la numeraciï¿½n de las fï¿½rmulas a la derecha
\usepackage[T1]{fontenc}
\usepackage{graphicx}
\usepackage[all]{xy}
\usepackage{enumerate}
\usepackage{array}
\usepackage{amssymb}
\usepackage[cp1252]{inputenc}
\usepackage{amscd}
\usepackage{hyperref}
\theoremstyle{plain}
\newtheorem{prop}{Proposition}[section]
\newtheorem{teo}[prop]{Theorem}
\newtheorem{cor}[prop]{Corollary}
\newtheorem{lema}[prop]{Lemma}
\newtheorem{defn}[prop]{Definition}
\newtheorem{obs}[prop]{Remark}
\theoremstyle{definition}
\newtheorem{ejem}{Example}

\theoremstyle{remark}

\numberwithin{equation}{section}

\newcommand{\St}{St(x,\Gamma_{n})}
\newcommand{\ef}{\mathbf{\Gamma}} %estructura fractal

\begin{document}

\title[Fractal Dimension for Fractal Structures]{Fractal Dimension for Fractal Structures}
\author{M.A Sánchez-Granero and M. Fernández-Martínez}
\curraddr{Area of Geometry and Topology \\ Faculty of
Science \\
Universidad de Almer\'{\i}a \\ 04071 Almer\'{\i}a \\ Spain}
\email{\url{misanche@ual.es} \and \url{fmm124@ual.es}}
\thanks{The first author acknowledges the support of the Spanish Ministry
of Science and Innovation, grant MTM2009-12872-C02-01.}



%\date{\today}

%\dedicatory{}

\begin{abstract}
The main goal of this paper has a double purpose. On the one hand,
we propose a new definition in order to compute the fractal dimension of a subset respect to any fractal structure, which completes the theory of classical box-counting dimension. Indeed, if we select the so called \emph{natural fractal structure} on each euclidean space, then we will get the box-counting dimension as a particular case. Recall that box-counting dimension could be calculated over any euclidean space, although it can be defined over any metrizable one. Nevertheless, the new definition we present can be computed on an easy way over any space admitting a fractal structure. Thus, since a space is metrizable if and only if it supports a starbase fractal structure, our model allows to classify and distinguish a much larger number of topological spaces than the classical definition.\newline

On the other hand, our aim consists also of studying some applications of effective calculation of the fractal dimension over a kind of contexts where the box-counting dimension has no sense, like the domain of words, which appears when modeling the streams of information in Kahn's parallel computation model. In this way, we show how to calculate and understand the fractal dimension value obtained for a language generated by means of a regular expression, and also we pay attention to an empirical and novel application of fractal dimension to natural languages.
%This paper provides two approaches for a definition of fractal dimension for a fractal structure. We prove that both approaches generalize the box-counting dimension, which enables the use of a larger number of fractal structures in order to compute the fractal dimension than the classical one allows. Furthermore, we study the fractal dimension with respect to a natural fractal structure defined in any self-similar set, and we also show the easiness of its effective calculation.\newline
%
%We also study some applications of this new tool to the effective calculation of the fractal dimension. In this way, we consider a suitable fractal structure in the context of a domain of words, and we show how the fractal dimension of any language can be calculated and understood, and furthermore, we also present the effective calculation of the fractal dimension of a language generated by means of a regular expression.

\end{abstract}

\maketitle
\parskip 2ex

% ----------------------------------------------------------------
\section{Introduction}
Since the concept of fractal was introduced by Benoît Mandelbrot at early seventies, the study and analysis of this kind of non-linear objects has become more and more important. In this way, fractals have been applied to a diverse spectrum of fields in science, such as the diagnosis of diseases (like osteoporosis (\cite{RUT92}) or cancer (\cite{BAI00}), dynamical systems (\cite{COL10}),
ecology (\cite{BAR03}), earthquakes (\cite{HIR89}), detection of
eyes in human face images (\cite{LIN01}), and
the analysis of the human retina (\cite{LAN03}), just to name a few.
Furthermore, topology allows the study of fractals from both theoretical and applied points of view, by means of fractal structures, which were first sketched in \cite{BA92}.
In this way, the introduction of fractal structures has allowed to formalize some topics on fractal theory, and its use leads to connect diverse and interesting subjects on general topology like transitive quasi-uniformities, non-archimedean quasimetrization, metrization, topological and fractal dimensions, self-similar sets and even space filling curves (see \cite{SG10}).
Indeed, one of the main tools applied to the study of fractals is the fractal dimension, understood as the classical \emph{box-counting dimension}, since it is a single quantity which offers some information about the complexity of a given set.

One of the purposes of this paper consists of providing a new definition in order to calculate the fractal dimension of a set respect to any fractal structure. With this in mind, we examine some properties of that definition and relate it with the box-counting dimension where the latter can be defined, by providing some conditions over the elements of the fractal structure we select. In this way, we also study the new fractal dimension over self-similar sets, which constitute a kind of fractals which always present a fractal structure in a natural way.

Recall that a fractal structure is a countable collection of
coverings which constitutes an approximation of the whole space by a discrete sequence of
levels. Thus, it is the perfect place to provide a definition of fractal dimension.
The first notion of dimension we propose depends only on the fractal structure we select but not on any metric. Notice that it may seems counterintuitive at a first glance, since spaces as the Cantor set have dimension one with respect to some fractal structure, but this fact is a consequence of depending only on the fractal structure and not on any metric we can consider on the space. We explain this situation in remark \ref{obs:5}.
The second notion of dimension we propose depends not only on the fractal structure but also on a metric or a distance function.
In this way, the second version of fractal dimension agrees with the first one when working with the semimetric associated with the fractal structure. Thus, we can use the second notion if we need to consider the size of the elements of each level of the fractal structure.

On the other hand, one of the main advantages of box-counting dimension consists of the easiness of its effective calculation and empirical estimation over the euclidean spaces, though it can also be defined over the metrizable ones. However, this classical definition does not have so good analytical properties as other definitions present, like the Hausdorff dimension. Moreover, our new definitions in order to compute the fractal dimension of a subset respect to any fractal structure can be easily calculated over any space admitting a fractal structure. In this way, recall that a topological space is metrizable if and only if it admits a starbase fractal structure (\cite[Theorem 4.2]{SG02}), so that the new definitions allow to classify and distinguish a larger volume of topological spaces than the box-counting model.
Note that the box-counting dimension is just a particular case of the new fractal dimension definitions, since it suffices with taking the \emph{natural fractal structure} on any euclidean space which we formally introduce in section \ref{sec:3}.

The other main goal that this paper provides consists of computing the fractal dimension in interesting contexts where the box-counting dimension has no sense, as the domain of words which appears when modeling the streams of information in Kahn's parallel computation model (see \cite{KAH74} and \cite{MAT94}). Thus, we calculate and explain the fractal dimension of a language generated by means of a regular expression where infinite length words could exist. Moreover, we present another empirical application of fractal dimension to the study of the fractal complexity of any natural language, which allows to compare natural languages, and even to analyze the variety of words used in any text or book written in any language. It is also possible to quantify the complexity of a translation of any text respect to its original version. Finally, we show how fractal dimension is a useful tool in order to study the efficiency of an encoding system as the standard BCD.

%On the other hand, both fractal dimensions are easy to calculate and can be used in practical applications, even in those were the box-counting dimension has no sense. We give an example of this by showing how to calculate the fractal dimension of a language as a subspace of the domain of words.

\section{Preliminaries}
Let's start with some preliminary topics.

\subsection{Fractal structures, self-similar sets and quasipseudometrics}\label{sec:1}
The main purpose of this section consists of recalling some notations and basic notions that will be useful in this paper.

In this way, the key concept we are going to use is about \emph{fractal structures}. Nevertheless, although a more natural use of them is in the study of fractals, and in particular self-similar sets (see \cite{SGnd}), its introduction was first motivated in order to characterize non-archimedeanly quasimetrization (see \cite{SG99}).
%Nevertheless, a more natural use of such structures is in the study of fractal sets, and in particular, self-similar sets (see \cite{SGnd}).
The use of fractal structures provides a powerful tool in order to study new models for a fractal dimension definition, since they will allow to distinguish and classify a larger volume of spaces than by using the classical definitions of fractal dimension (which will be obtained as a particular case), that only work over metrizable spaces.
%Indeed, the use of fractal structures when defining a fractal dimension will allow to distinguish and classify a larger volume of spaces,
So that, these kind of topological spaces constitutes a perfect place in order to develop a theory on fractal dimension.

%In this paper, all topological spaces are assumed to be $T_{0}$.\newline
%Next, we introduce some notations and recall basic notions.
%In this way, we advance the key concept we are going to use is about the generalized fractal-spaces or GF-spaces for short, whose utilization while defining fractal dimension will allow us to calculate this significant quantity relative to any subset of a topological space. Taking it into account, we can get a more general model in order to determine the fractal dimension of a set that working only over the euclidean or metrizable spaces where box-counting and Hausdorff dimensions have been defined.

Let $\Gamma$ be a covering of $X$. Thus, we will denote $St(x,\Gamma)=\cup\{A\in
\Gamma:x\in A\}$ and $U_{x \Gamma}=X \setminus \cup \{A\in
\Gamma:x\notin A\}$.
Furthermore, if $\ef=\{\Gamma_{n}:n\in \mathbb{N}\}$ is a countable family of coverings of
$X$, then we will denote $U_{xn}=U_{x
\Gamma_n}$, $\mathcal{U}_{x}^{\ef}=\{U_{xn}:n\in \mathbb{N}\}$
and $St(x,\ef)=\{\St:n\in \mathbb{N}\}$ .

The next definition was introduced in \cite{SG99}.
\begin{defn}
Let $X$ be a topological space. A pre-fractal structure on $X$ is a countable
family of coverings (called levels) $\ef=\{\Gamma_{n}:n\in \mathbb{N}\}$ such
that $\mathcal{U}_{x}^{\ef}$  is an open neighborhood
base of every point $x\in X$.\newline
Moreover, if $\Gamma_{n+1}$ is a refinement of $\Gamma_{n}$ (which can be denoted by $\Gamma_{n+1}\preceq \Gamma_{n}$), such that for all $x\in A$ with $A\in \Gamma_{n}$, there exists $B\in \Gamma_{n+1}$ such that $x\in B\subseteq A$, we will say that $\ef$ is a fractal structure on $X$.\newline
If $\ef$ is a (pre-) fractal structure on $X$, then we will say that $(X,\ef)$
is a generalized (pre-) fractal space, or simply a (pre-) GF-space. If there is
no doubt about $\ef$, then we will say that $X$ is a (pre-) GF-space.
\end{defn}

\begin{obs}
Note that the levels we use in order to define a fractal structure $\ef$ are not coverings in the usual sense, since we are going to enable the possibility that could exist elements on any level of the fractal structure which can appear two times or more when determining the whole family $\ef$. For instance, $\Gamma_{1}=\{[0,\frac{1}{2}],[\frac{1}{2},1],[0,\frac{1}{2}]\}$ could be the first level of a given fractal structure $\ef$ defined over the closed unit interval.
\end{obs}

Note also that if $\ef$ is a pre-fractal structure, then any of its levels is a closure preserving closed covering for each(see [\cite{SG02B},Prop. 2.4]).

%CREO QUE NO ES ESTRICTAMENTE NECESARIO
%*****************************************
%On the other hand, let $\Gamma$ be a covering on a topological space $X$ and let $x\in X$. The
%order of $x$ in $\Gamma$ is defined as the number of elements of the set $\{A\in
%\Gamma: x\in A\}$ minus $1$, and then, the order of the covering $\Gamma$ is
%defined as $Ord(\Gamma)=\sup\{Ord(x,\Gamma):x\in X\}$.

%Thus, if $(X,\ef)$ is a GF-space, we define
%$Ord(\ef)=\sup\{Ord(\Gamma_{n}):n\in \mathbb{N}\}$.
%*****************************************

If $\ef$ is a fractal structure on $X$ and $St(x,\ef)$ is a neighborhood base of $x$ for all $x\in X$, we will call $\ef$ a \emph{starbase} fractal structure.
%A tiling over a topological space $X$ is a covering of $X$ by sets which are the closures of their pairwise-disjoint interiors. A fractal structure $\ef$ is said to be a tiling over $X$, if $\Gamma_{n}$ is a tiling over $X$ for all $n\in \mathbb{N}$, namely, if $A_{n}$ is the closure of the interior of $A_{n}$, for all $A_{n}\in \Gamma_{n}$, and furthermore, the sets $A_{n}\in \Gamma_{n}$ are pairwise-disjoint.
%A fractal structure $\ef$ is said to be finite
%if all levels $\Gamma_{n}$ are finite coverings.
%A fractal structure $\ef$ is
%said to be locally finite if for each level $\Gamma_{n}$ of the fractal
%structure $\ef$, we have that every point $x\in X$ belongs to a finite number
%of elements $A\in \Gamma_{n}$.
In general, if $\Gamma_{n}$ has the property $P$ for all $n\in \mathbb{N}$, and
$\ef=\{\Gamma_{n}:n\in \mathbb{N}\}$ is a fractal structure on $X$, we will say that $\ef$ is a fractal
structure with the property $P$, and that $(X,\ef)$ is a GF-space with
the property $P$.
For instance, if $\Gamma_{n}$ is a finite covering for all natural number $n$ and $\ef$ is a fractal structure on $X$, then we will say that $\ef$ is a finite fractal structure on $X$, and that $(X,\ef)$ is a finite GF-space.

%The GF-space product was also introduced in \cite{SG99}. In fact, let $\{(X_{i},\ef^{i}):i\in \mathbb{N}\}$ be a countable family of GF-spaces. The fractal structure which is induced by the countable product of the previous GF-spaces is defined as $\prod_{i\in \mathbb{N}}\ef^{i}=\{\Gamma_{n}:n \in \mathbb{N}\}$, where $\Gamma_{n}=\big\{\bigcap_{i\leq n}p_{i}^{-1}(A_{n}^{i}): A_{n}^{i}\in \Gamma_{n}^{i}\big\}$ for all $n\in \mathbb{N}$. Here, $p_{i}:\prod_{k\in \mathbb{N}}(X_{k},\ef^{k})\rightarrow (X_{i},\ef^{i})$ is the projection mapping over the GF-space on the $i$-position. Then the product GF-space will be denoted as $(\prod_{i\in \mathbb{N}}X_{i},\prod_{i\in \mathbb{N}}\ef^{i})$. In particular, if $\ef^{i}=\ef^{j}$, for all $i,j\in \{1,\ldots, k\}$, then we define $\ef^{k}=\ef\times \overset{k)}{\ldots} \times \ef$ in a natural way.
%\newline

On the other hand, we also recall the definition of self-similar set provided by Hutchinson (see \cite{HUT81}).

\begin{defn}
Let $I=\{1,\ldots, m\}$ be a finite index set and let $\{f_{i}:i\in I\}$ be a family of contractive mappings defined from a complete metric space $X$ into itself.
Then there exists a unique non-empty compact subset $K$ of $X$ such that $K=\cup_{i\in I}f_{i}(K)$, which is called a self-similar set.
\end{defn}
In classical non-linear theory, $(X,\{f_{i}:i\in I\})$ is called an iterated function scheme (which we will denote by IFS for short), and the self-similar set $K$, as the \emph{atractor} of that IFS.
Next, we provide an interesting example which describes analytically the so called Sierpinski's gasket, which is a typical example of a strict self-similar set.

\begin{ejem}\label{ejem:1}
Let $I=\{1,2,3\}$ be a finite index set and let $\{f_{i}:i\in I\}$ be a finite set of similarities over the euclidean plane which are defined by
\begin{equation}\label{eq:48}
f_{i}(x,y)=\left\{
\begin{array}{ll}
\hbox{$(\frac{x}{2},\frac{y}{2})$}\ \hfill \text{if} & \hbox{$i=1$} \\
\hbox{$f_{1}(x,y)+(\frac{1}{2},0)$}\ \hfill\text{if} & \hbox{$i=2$} \\
\hbox{$f_{1}(x,y)+(\frac{1}{4},\frac{1}{2})$}\ \hfill\text{if} & \hbox{$i=3$}
\end{array}
\right.
\end{equation}
for all $(x,y)\in \mathbb{R}^{2}$. Thus, the Sierpinski's gasket is determined on an unique way as the non-empty compact subset verifying the next Hutchinson's equation:
%\begin{equation}
$K=\cup_{i\in I}f_{i}(K)$. In this way,
%\end{equation}
note that each component $f_{i}(K)$ is a self-similar copy of the atractor of the IFS $(\mathbb{R}^{2},\{f_{i}:i\in I\})$.
\end{ejem}
To consult a procedure in order to approach self-similar sets, see \cite{EA03} and \cite{CHI10}.
Self-similar sets constitute an interesting kind of fractals that are characterized by having a
fractal structure in a natural way, which was first sketched in \cite{BA92}. Indeed, that paper becomes the origin of the term \emph{fractal structure}.
Next, we present the description of such fractal structure
(see \cite{SGnd}).

\begin{defn}
Let $I=\{1,\ldots, m\}$ be a finite
index set, and let $(X,\{f_{i}:i\in I\})$ be an IFS whose associated
self-similar set is $K$. The natural fractal structure on
$K$ can be defined as the countable family of coverings $\ef=\{\Gamma_{n}:n\in
\mathbb{N}\}$, where $\Gamma_{n}=\{f_{\omega}^{n}(K):\omega\in I^{n}\}$ for every
natural number $n$. Here for all $n\in \mathbb{N}$ and all
$\omega=\omega_{1} \ \omega_{2}\ \ldots \ \omega_{n}\in I^{n}$, we denote
$f_{\omega}^{n}=f_{\omega_{1}}\ \circ \ \ldots \ \circ \ f_{\omega_{n}}$.
\end{defn}

\begin{obs}
Another available description for this fractal structure is as follows:
$\Gamma_{1}=\{f_{i}(K):i\in I\}$ and $\Gamma_{n+1}=\{f_{i}(A): A\in \Gamma_{n},
i\in I\}$ for all $n\in \mathbb{N}$.
\end{obs}

On example \ref{ejem:1} we described analytically the IFS whose associated
self-similar set is the Sierpinski's triangle. Next, we are going to present
the natural fractal structure associated with this strict self-similar set.
\begin{ejem}
The natural fractal structure associated with the Sierpinski's
triangle can be described as the countable family of coverings
$\ef=\{\Gamma_{n}: n\in \mathbb{N}\}$, where $\Gamma_{1}$ is the union of three
equilateral ``triangles'' with sides equal to $\frac{1}{2}$, $\Gamma_{2}$
consists of the union of $3^{2}$ equilateral ``triangles'' with sides equal to
$\frac{1}{2^{2}}$, and in general, $\Gamma_{n}$ is the union of $3^{n}$
equilateral ``triangles'' whose sides are equal to $\frac{1}{2^{n}}$.
Furthermore, this is a finite starbase %tiling
fractal structure with
$Ord(\ef)=1$.
\end{ejem}

%\begin{center}
%\begin{figure}[h]\label{fig:2}
% \centering
%   \includegraphics[width=50mm,height=40mm]{sierpins.eps}
%   \caption{Natural fractal structure associated with the Sierpinski's gasket.}
% \label{fig:ejemplo}
%\end{figure}
%\end{center}

Moreover, recall that a \emph{quasipseudometric} on a set $X$ is a non-negative
real-valued function $d$ on $X\times X$ such that for all $x,y,z\in X$, verifies
the following two conditions: (i) $d(x,x)=0$ and (ii) $d(x,y)\leq d(x,z)+d(z,y)$. If in
addition $d$ satisfies also the next one: (iii) $d(x,y)=d(y,x)=0$ iff $x=y$,
then $d$ is called a \emph{quasi-metric}. In particular, a \emph{non-archimedean
quasipseudometric} is a quasipseudometric which also verifies that $d(x,y)\leq
\max\{d(x,z),d(z,y)\}$ for all $x,y,z\in X$. Further, we have that each quasipseudometric $d$ on $X$ generates a quasiuniformity $\mathcal{U}_{d}$
on $X$ which has as a base the family of sets of the form $\{(x,y)\in X\times X:
d(x,y)<2^{-n}\}$, $n\in \mathbb{N}$. Then the topology $\tau(\mathcal{U}_{d})$
induced by $\mathcal{U}_{d}$ will be denoted simply by $\tau(d)$.
Therefore, a space $(X,\tau)$ is said to be (non-archimedeanly) quasipseudometrizable if
there exists a (non-archimedean) quasipseudometric $d$ on $X$ such that
$\tau=\tau(d)$. Recall that the theory of quasiuniform spaces is covered in \cite{FLE82}.

%\begin{center}
%\begin{figure}[h]\label{fig:2}
% \centering
%   \includegraphics[width=40mm,height=40mm]{sierpins.eps}
%   \caption{The fractal structure associated with Sierpinski's gasket in a natural way is a very good one.}
% \label{fig:ejemplo}
%\end{figure}
%\end{center}

\subsection{The classical definition of fractal dimension over euclidean spaces}\label{sec:2}

Fractal dimension is one of the main tools used in order to study fractals, since it is a single value which provides information about its complexity and the irregularities it presents when being examined with enough level of detail. In this way, fractal dimension is usually understood as the classical box-counting dimension.
The latter has been also known as \emph{information dimension, Kolmogorov entropy, capacity dimension, entropy dimension, metric dimension , \ldots} etc.
Note that the box-counting definition is \emph{better} from an applied point of view than other theoretical ones, like Hausdorff dimension, since the easiness of its effective calculation an empirical estimation. Nevertheless, box-counting dimension have not so good theoretical properties as Hausdorff dimension, for instance. The basic theory on box-counting dimension can be found in \cite{FAL90}.
Thus, the (lower/upper) box-counting dimensions of a subset $F\subset \mathbb{R}^{d}$ are given by the following (lower/upper) limit:
\begin{equation}\label{eq:81}
\dim_{B}(F)=\lim_{\delta\rightarrow 0}\frac{\log N_{\delta}(F)}{-\log \delta}
\end{equation}
where $\delta$ is the scale used in the study of $F$ and $N_{\delta}(F)$ can be calculated on an equivalent way as one of the following quantities (see \cite[Equivalent definitions 3.1]{FAL90}):
\begin{enumerate}
%NO UTILIZAMOS LAS DOS PRIMERAS
%\item the smallest number of closed balls of radius $\delta$ that cover $F$.
%\item the smallest number of cubes with sides equal to $\delta$ that cover $F$.
\item the number of $\delta-$cubes that meet $F$. Recall that a $\delta-$cube in $\mathbb{R}^{d}$ is a set of the form $[k_{1}\ \delta,(k_{1}+1)\ \delta]\times \ldots \times [k_{d}\ \delta,(k_{d}+1)\ \delta]$ where $k_{i}$ are integers for all $i\in \{1,\ldots, d\}$. \label{d:3}
\item the number of $\delta=\frac{1}{2^{n}}-$cubes that intersect $F$, with $n\in \mathbb{B}$. \label{d:4}
\item the smallest number of sets of diameter at most $\delta$ that cover $F$. \label{d:5}
\item the largest number of disjoint balls of radius $\delta$ centered on $F$. \label{d:6}
\end{enumerate}
Note also that the limit described at \ref{eq:81} can be discretized by means of $\delta=\frac{1}{2^{n}}$.
In this way, $N_{\delta}(F)$ is just the number of elements of each level $\Gamma_{n}$ of the fractal structure which meet $F$. Moreover, the box-counting dimension can be estimated as the slope of a log-log graph plotted over a suitable discrete collection of scales $\delta$.\newline
Hence a natural idea arises: we can propose a new definition of fractal dimension for any fractal structure which generalizes the classical box-counting dimension and allows to classify and distinguish a larger volume of spaces than by means of the classical box-counting dimension definition. Thus, it also results useful in order to calculate the fractal dimension over another kind of spaces, such as the non-euclidean ones, where the box-counting dimension can have no sense. Thus, we show some interesting applications of this fact on section \ref{sec:6}, where the box-counting dimension cannot been applied in the context of domains of words. Nevertheless, our new definition in order to compute the fractal dimension of a subset allows to calculate and illustrate it.\newline
On the other hand, the next remark results useful in this paper.
\begin{obs}\label{rem:F}
In order to calculate the (lower/upper) box-counting dimensions of any subset
$F$ of a space $X$, it suffices with taking limits as $\delta\rightarrow 0$ by
means of a decreasing sequence $\{\delta_{n}\}_{n\in \mathbb{N}}$ verifying that
$\delta_{n+1}\geq c\ \delta_{n}$ for all $n\in \mathbb{N}$, where $c\in (0,1)$ is a
suitable constant.
\end{obs}

\section{Generalized definition of fractal dimension for a subset respect to any fractal structure}\label{sec:3}

The starting point of the fractal dimension theory for fractal structures begins by
taking into account the euclidean metrizable space $\mathbb{R}^{d}$ as well as
the equivalent definition \ref{d:4} for the quantity $N_{\delta}(F)$ that
appears at the box-counting introduction seen in subsection \ref{sec:2}.
Indeed, we can construct a fractal structure $\ef$ on any euclidean space
$\mathbb{R}^{d}$ in a natural way verifying some desirable properties: it is a locally finite
tiling starbase fractal structure which has a finite order. Indeed, it is going to be the so called natural fractal srtucture over any euclidean space, whose definition is as follows.
\begin{defn}
The natural fractal structure on the euclidean space $\mathbb{R}^{d}$ is
defined as the countable family of coverings $\ef=\{\Gamma_{n}:n\in
\mathbb{N}\}$, whose levels are given by
$\Gamma_{n}=\{[\frac{k_{1}}{2^{n}},\frac{k_{1}+1}{2^{n}}]\times
[\frac{k_{2}}{2^{n}},\frac{k_{2}+1}{2^{n}}]\times \ldots \times
[\frac{k_{d}}{2^{n}},\frac{k_{d}+1}{2^{n}}]:k_{i}\in \mathbb{Z},i\in \{1,\ldots,d\}\}$ for all $n\in \mathbb{N}$.
\end{defn}
Note that the natural fractal structure $\ef$ on the euclidean space
$\mathbb{R}^{d}$ is just the tiling consisting of
$\frac{1}{2^{n}}-$cubes for all natural number $n$ on
$\mathbb{R}^{d}$ which allows to calculate the box-counting dimension of any
subset of such space. In this way, we are going to propose a new definition of
fractal dimension on a generic GF-space $(X,\ef)$ which generalizes the
box-counting one if we use the natural fractal
structure on the euclidean space $\mathbb{R}^{d}$. However, as
happens with box-counting dimensions, our model for a generic fractal dimension
is going to be defined using upper and lower limits. Nevertheless, one of the
main advantages it presents, consists of the easiness of its effective
calculation and theoretical interpretation. Taking it into account, we can consider a larger volume of fractal structures than box-counting dimension, in order to calculate the fractal dimension of a given set.
%In order to avoid possible confusion, in this paper we are going to use the notation $\dim_{B}$ and $\dim_{\ef}^{1}$ to denote box-counting dimension and fractal dimension I respectively.
%Also, for all natural number $n$ and all subset $F\subset X$, consider the next family of elements on each level of a given fractal structure $\ef$:
The key concept of this section is defined below.
\begin{defn}\label{def:1}
Let $(X,\ef)$ be a GF-space  and let $N_{n}(F)$ be
the number of elements of $\mathcal{A}_{n}(F)$ for all $n\in \mathbb{N}$. Thus, the (lower/
upper) fractal dimensions I of a non-empty bounded subset $F$ of $X$ respectively, are defined as the (lower/upper) limit:
%\begin{equation}
%\underline{\dim}_{\ef}^{1}(F)=\varliminf_{n\rightarrow \infty}\frac{\log N_{n}(F)}{n\log 2}
%\end{equation}
%\begin{equation}
%\overline{\dim}_{\ef}^{1}(F)=\varlimsup_{n\rightarrow \infty}\frac{\log N_{n}(F)}{n\log 2}
%\end{equation}
%and if those expressions agree, we will say that their common value is the fractal dimension I of $F$, namely
\begin{equation}
\dim_{\ef}^{1}(F)=\lim_{n\rightarrow \infty}\frac{\log N_{n}(F)}{n\log 2}
\end{equation}
if this limit exists, where $\mathcal{A}_{n}(F)$ is the next family of elements on each level of the fractal structure $\ef$, for all natural number $n$:
\begin{equation}\label{eq:83}
\mathcal{A}_{n}(F)=\{A\in \Gamma_{n}:A\cap F\neq \emptyset\}
\end{equation}
\end{defn}
%Accordingly, to calculate the fractal dimension I of a subset $F\subset X$, it
%suffices with counting the number of elements on the level $\Gamma_{n}$ of the fractal structure $\ef$ which meet
%$F$ for each natural number $n$. Another advantage it presents
%with respect to the box-counting definition, consists of the fact that we have already restricted the family $\ef$ to a countable collection of coverings of the whole space
%%only need a countable family of coverings,
%while box-counting definition uses a continuum
%range of scales. Moreover,
In order to estimate the fractal dimension from a computational point of view,
we can calculate the slope of the regression line on a log-log graph, just like with the box-counting dimension estimation.\newline
%Thus, our fractal dimension I can be calculated
%on an easier way than using box-counting dimension.
%One of the first properties we can show for fractal dimension I consists of generalizing
%the box-counting dimension of every subset of $\mathbb{R}^{d}$. In fact, it suffices with taking into account the natural fractal structure on $\mathbb{R}^{d}$.\newline
%Furthermore, it is possible to describe the natural fractal structure on every euclidean space $\mathbb{R}^{n}$, using the countable product of the GF-space $(\mathbb{R},\ef)$ with itself, while considering $\ef$ as the natural fractal structure on $\mathbb{R}$. Indeed, we can consider the following
%\begin{obs}
%Let $\mathbb{R}^{d}$ be an euclidean space, and let $F\subset \mathbb{R}^{d}$. The natural fractal structure on $\mathbb{R}^{d}$ can be described as $\ef^{d}$, where $\ef=\{\Gamma_{n}:n\in \mathbb{N}\}$ is the countable family of coverings of $F$ such that $\Gamma_{n}=\{[\frac{k}{2^{n}},\frac{k+1}{2^{n}}]:k\in \mathbb{Z}\}$ for all $n\in \mathbb{N}$. Here, we obtain the corresponding associated product GF-space $(\mathbb{R}^{d},\ef^{d})$. \hfill $\cqd$
%\end{obs}
In this way, the next result allows to calculate the box-counting dimension of any
subset $F\subset \mathbb{R}^{d}$ by means of the fractal dimension I formula: it
suffices with counting the number of $\frac{1}{2^{n}}$-cubes for all $n\in \mathbb{N}$, just like the
box counting dimension with equivalent definition (\ref{d:4}) for $N_{\delta}(F)$. Thus, the proof of the following result becomes straightforward.
\begin{teo}\label{prop:1}
Let $\ef$ be the natural fractal structure on the euclidean space
$\mathbb{R}^{d}$ and let $F$ be a subset of $\mathbb{R}^{d}$. Then the (lower/upper) box-counting dimension and the (lower/upper) fractal dimension I of $F$ are equal.
\end{teo}
Hausdorff dimension constitutes the reference of any definition of fractal dimension. In this way, we can check some of its analytical properties for the fractal dimension I definition (see \cite[Chapter 3]{FAL90}). Indeed, we have the following proposition.

\begin{prop}\label{prop:6}
Let $(X,\ef)$ be a GF-space. Then,
\begin{enumerate}
\item Both $\underline{\dim}^{1}_{\ef}$ and $\overline{\dim}^{1}_{\ef}$ are monotonic. \label{it:1}
\item $\overline{\dim}^{1}_{\ef}$ is finitely stable. \label{it:2}
\item Neither $\underline{\dim}^{1}_{\ef}$ nor $\overline{\dim}^{1}_{\ef}$ is countably stable. \label{it:3}
\item There exists countable sets $F\subset X$ such that $\dim_{\ef}^{1}(F)\neq 0$. \label{it:4}
\item There exists a locally finite tiling starbase fractal structure $\ef$ with \label{it:5}
finite order on a suitable space $X$ such that
$\dim_{\ef}^{1}(F)\neq \dim_{\ef}^{1}(\overline{F})$ for a given subset
$F\subset X$.
\end{enumerate}
\end{prop}

\begin{proof}
The more straightforward items are left to the reader.
\begin{enumerate}
\item[]{(4)} Recall that the countable stability property for a dimension function $\dim$ means that
%\begin{equation}\label{eq:72}
$\dim(\cup_{i\in I}F_{i})=\sup_{i\in I}\dim(F_{i})$,
%\end{equation}
where $\{F_{i}\}_{i\in I}$ is a countable family of subsets of $X$.
%We affirm that the property which appears on equation \ref{eq:72} is going to be false for fractal dimension I.
%In fact, taking into account the monotony of $\underline{\dim}^{1}_{\ef}$, it is clear that $\underline{\dim}^{1}_{\ef}(F_{j})\leq \underline{\dim}^{1}_{\ef}(\cup_{i\in I}F_{i})$ for all $j\in I$, so that, $\sup_{i\in I}\underline{\dim}^{1}_{\ef}(F_{i})\leq \underline{\dim}^{1}_{\ef}(\cup_{i\in I}F_{i})$. The case for upper limits may be done on a similar way.
%Now, we present the following counterexample in order to show that the other inequality is false.
Therefore, consider $X=[0,1]$ with $F=\mathbb{Q}\cap [0,1]$, and let $\ef$ be the natural fractal structure on $\mathbb{R}$ induced in $[0,1]$, which can be described by $\Gamma_{n}=\{[\frac{k}{2^{n}},\frac{k+1}{2^{n}}]: k\in \{0, 1, \ldots, 2^{n}-1\}\}$ for all $n\in \mathbb{N}$. Then, it is clear that $N_{n}(F)=2^{n}$, so that $\dim^{1}_{\ef}(F)=1$.
Further, note that (4) implies (3). \newline
%Furthermore, $\dim^{1}_{\ef}(\{q_{i}\})=0$, for all $q_{i}\in F$.
%Accordingly, it is not in general true that $\underline{\dim}^{1}_{\ef}(\cup_{i\in I}F_{i})\leq \sup_{i\in I}\underline{\dim}^{1}_{\ef}(F_{i})$. The same counterexample is valid for the case of upper fractal dimension I.
%\item If $F=\mathbb{Q}\cap [0,1]$ and $\ef$ is the previous fractal structure on $[0,1]$, then $\dim_{\ef}^{1}(F)=1>0$.
\item[]{(5)} Indeed, consider the fractal structure $\ef=\{\Gamma_{n}:n\in \mathbb{N}\}$ with $\Gamma_{n}=\{[\frac{k}{2^{n}},\frac{k+1}{2^{n}}]\times \{0\}:k\in \{0,1,\ldots, 2^{n}-1\}\}\cup \{\{\frac{1}{2^{m}}\}\times [\frac{k}{2^{n}},\frac{k+1}{2^{n}}]: k\in \{0,1,\ldots, 2^{n}-1\}, m\in \mathbb{N}\}$ for all natural number $n$, on the space $X=([0,1]\times \{0\})\cup \{\{\frac{1}{2^{n}}\}\times [0,1]: n\in \mathbb{N}\}$. Take also $F=\bigcup_{k\in \mathbb{N}}(\frac{1}{2^{k+1}},\frac{1}{2^{k}})\times \{0\}$ as a subset of $X$. Note that $Ord(\ef)=2$. On the other hand, it is clear that $\overline{F}=[0,1]\times \{0\}$, so that $N_{n}(F)=2^{n}$ and $N_{n}(\overline{F})=\infty$ for all natural number $n$, which implies that $\dim_{\ef}^{1}(F)=1$ and $\dim_{\ef}^{1}(\overline{F})=\infty$.
%GF-space $(X,\ef)$, where $X=([0,1]\times \{0\})\cup \{\{\frac{1}{2^{n}}\}\times [0,1]: n\in \mathbb{N}\}$ and the fractal structure $\ef$ is given as the countable family of coverings $\ef=\{\Gamma_{n}:n\in \mathbb{N}\}$ with $\Gamma_{n}=\{[\frac{k}{2^{n}},\frac{k+1}{2^{n}}]\times \{0\}:k\in \{0,1,\ldots, 2^{n}-1\}\}\cup \{\{\frac{1}{2^{m}}\}\times [\frac{k}{2^{n}},\frac{k+1}{2^{n}}]: k\in \{0,1,\ldots, 2^{n}-1\}, m\in \mathbb{N}\}$. Note that $Ord(\Gamma_{n})=2$ for all $n\in \mathbb{N}$. Consider $F=\bigcup_{k\in \mathbb{N}}(\frac{1}{2^{k+1}},\frac{1}{2^{k}})\times \{0\}$. Thus, it is clear that $\overline{F}=[0,1]\times \{0\}$, so that we have $N_{n}(F)=2^{n}$ and $N_{n}(\overline{F})=\infty$, which implies that $\dim_{\ef}^{1}(F)=1$ and $\dim_{\ef}^{1}(\overline{F})=\infty$.
\end{enumerate}
\end{proof}
%As seen on theorem \ref{prop:1}, we have that fractal dimension I agrees with box-counting dimension on the euclidean space $\mathbb{R}^{d}$ equipped with its natural fractal structure $\ef$. Further, note that box-counting dimension can be defined also on a metrizable space.
%%Furthermore, our fractal dimension I method looks easier to calculate than box-counting one.
%On the other hand, it is possible that sometimes we only need a suitable
%estimation for the box-counting dimension of a given subset. In order to obtain
%it,
The next question we are going to investigate consists of the possibility of
determining an approximation of the box-counting dimension of a subset
on a metric space, in terms of its fractal dimension I.
In this way, we define the diameter of each level of a fractal structure $\ef$ as well as the diameter of a subset in a given level of $\ef$ as follows.
\begin{defn}\label{def:2}
Let $\ef$ be a fractal structure on a metric space $(X,\rho)$, and let $F$ be a
subset of $X$. Then, the diameter of each level $\Gamma_{n}$ of the fractal
structure $\ef$ is defined as
\begin{equation}
\delta(\Gamma_{n})=\sup\{\delta(A): A\in \Gamma_{n}\}
\end{equation}
and the diameter of $F$ on each level $\Gamma_{n}$ of the fractal structure $\ef$ is given as the quantity
\begin{equation}
\delta(F,\Gamma_{n})=\sup \{\delta(A):A\in \mathcal{A}_{n}(F)\}
\end{equation}
\end{defn}
Note that starbase fractal structures lead to GF-spaces with some desirable properties. Taking it as well as definition \ref{def:2} into account, we have found a condition over the diameters of each level of the fractal structure
%under consideration,
in order to get this kind of topological spaces.
%Indeed, it suffices with verifying that $\delta(\Gamma_{n})\rightarrow 0$ as $n\rightarrow \infty$. Note that this condition is natural over fractals as self-similar sets.
The proof of the next theorem results straightforward.
%Starbase fractal structures are going to be very useful in our fractal dimension theory. In fact, there is a natural sufficient condition we can require to a fractal structure $\ef$ in order to get an associated starbase GF-space $(X,\ef)$. Indeed, we only need that the sequence of diameters of each level $\Gamma_{n}$ of the fractal structure $\ef$ converges to $0$ when we explore the fractal set with sufficient detail. The proof is straightforward.
\begin{prop}\label{prop:7}
Let $\ef$ be a fractal structure on a compatible metric (or quasimetric) space $(X,\rho)$, and
suppose that $\delta(\Gamma_{n})\rightarrow 0$. Then
$\ef$ is starbase.
\end{prop}
 \begin{proof}
 Indeed, it suffices with checking that $St(x,\ef)$ is a neighborhood base of $x$ for all $x\in X$. First of all, it is clear that $x\in U_{xn}\subset St(x,\Gamma_{n})$ for all $x\in X$ and all natural number $n\in \mathbb{N}$, with $U_{xn}$ being a neighborhood of $x$ about the topology associated with the fractal structure $\ef$. On the other hand, let $x\in X$ be a fixed point on $X$, and let $\varepsilon>0$. Since $\delta(\Gamma_{n})\rightarrow 0$, then there exists a natural number $n_{0}\in \mathbb{N}$ such that $\delta(\Gamma_{n})<\varepsilon$ for all $n\geq n_{0}$. Hence, let $m\in \mathbb{N}$ be a natural number such that $m\geq n_{0}$, and consider $St(x,\Gamma_{m})$. Then, for all $y\in St(x,\Gamma_{m})$ there exists $A\in \Gamma_{m}$ with $x\in A$, such that $y\in A$. Moreover, as $\delta(\Gamma_{m})<\varepsilon$, we have that $\rho(x,y)<\varepsilon$, namely, $y\in B_{\rho}(x,\varepsilon)$. Accordingly, there exists $m\in \mathbb{N}$ such that $St(x,\Gamma_{m})\subset B_{\rho}(x,\varepsilon)$.
 \end{proof}
%The key condition we can require to a fractal structure $\ef$ in order to obtain a suitable estimation for the box-counting dimension of a given subset $F$ in terms of its fractal dimension I, consists of the fact that the sequence of the diameters on each level $\Gamma_{n}$ of the fractal structure $\ef$ decreases on a geometric way.
%Observe that the latter is a natural condition that a fractal structure can verify, as it happens while taking self-similar sets. Indeed, we have the following theorem.
A very natural condition which a fractal structure could verify, consists of the fact that the sequence of diameters of each level of the fractal structure $\ef$ decreases on a geometric way, which constitutes the main idea in the following theorem.
\begin{teo}\label{teo:1}
Let $\ef$ be a fractal structure on a metric space $(X,\rho)$ with $F$ being a subset of $X$, and suppose that there exists a constant $c\in (0,1)$ such that the next condition is verified:
\begin{equation}\label{eq:50}
\delta(F,\Gamma_{n+1})\leq c\ \delta(F,\Gamma_{n})
\end{equation}
for all natural number $n$. Take the constant $\gamma_{c}=\frac{-\log 2}{\log c}$. Then,
\begin{enumerate}
\item
%\begin{equation}\label{des:1}
$\overline{\dim}_{B}(F)\leq \gamma_{c} \cdot \overline{\dim}^{1}_{\ef}(F)$.
%\end{equation}
\item
%\begin{equation}\label{des:2}
$\underline{\dim}_{B}(F)\leq \gamma_{c} \cdot \underline{\dim}^{1}_{\ef}(F)$.
%\end{equation}
\item Moreover, if there exist both box-counting dimension and fractal dimension I of $F$, then
%\begin{equation}\label{ec:5}
$\dim_{B}(F)\leq \gamma_{c} \cdot \dim_{\ef}^{1}(F)$.
%\end{equation}
\end{enumerate}
\end{teo}
\begin{proof}
In order to calculate the box-counting dimension of $F$, let $N_{\delta}(F)$ be as (\ref{d:5}) on equivalent box-counting definition seen at preliminary subsection \ref{sec:2}.
%be the smallest number of sets of diameter at most $\delta$ that cover $F$.
\begin{enumerate}
\item Taking the geometric decreasing of the diameters of each level of the fractal structure $\ef$ into account, we affirm that there exists a constant $c\in (0,1)$ such that $\delta(F,\Gamma_{n})\leq c^{n-1}\cdot \delta(F,\Gamma_{1})$ for all $n\in \mathbb{N}$. Denote $\delta_{n}=c^{n-1}\cdot \delta(F,\Gamma_{1})$ as the general term of a decreasing sequence which converges to $0$. Thus,
%Using the expression \ref{eq:50}, we affirm that there exists a constant $c\in (0,1)$ such that $\delta(F,\Gamma_{n})\leq c^{n-1}\ \delta(F,\Gamma_{1})$. Note that $\delta_{n}=c^{n-1}\ \delta(F, \Gamma_{1})$ is the general term of a decreasing sequence which converges to $0$. Hence, we can establish what follows:
%\begin{equation}\label{eq:64}
%\begin{split}
$\overline{\dim}_{B}(F)
=\varlimsup_{n\rightarrow \infty} \frac{\log N_{\delta_{n}}(F)}{-\log \delta_{n}}
%&\leq \varlimsup_{n\rightarrow \infty}\frac{\log N_{n}(F)}{(1-n)\ \log c}\\
\leq \varlimsup_{n\rightarrow \infty}\frac{\log N_{n}(F)}{-n\ \log c}
=\gamma_{c}\cdot\ \varlimsup_{n\rightarrow \infty}\frac{\log N_{n}(F)}{n\ \log
2}$, where remark \ref{rem:F} is used in the first equality.
%\end{split}
%\end{equation}
%Note that we have taken into account the fact that $\delta_{n+1}\geq c\ \delta_{n}$ with $c\in (0,1)$, for all natural number $n$ in order to satisfy the first equality. Furthermore, it is clear that $\delta(A)\leq \delta_{n}$ for all $A\in \mathcal{A}_{n}(F)$.
%which implies the inequality that appears on \ref{eq:64}.
%$\delta(A)\leq \delta(F,\Gamma_{k})\leq c^{k-1}\ \delta(F,\Gamma_{1})=\delta_{k}$ for all $A\in \Gamma_{k}$ which meets $F$. Hence, we get the first inequality by the definition of $N_{\delta_{k}}(F)$.
\item Consider the previous sequence $\{\delta_{n}\}_{n\in \mathbb{N}}$.
Since
%\begin{equation}
%\nonumber
%\varliminf_{\delta\rightarrow 0}\frac{\log N_{\delta}(F)}{-\log \delta}=\inf\Bigg\{\varliminf_{n\rightarrow \infty}\frac{\log N_{\delta_{n}}(F)}{-\log \delta_{n}}: \delta_{n}\rightarrow 0 \ \text{as} \ n\rightarrow \infty\Bigg\}
%\end{equation}
%It is clear that
%\begin{equation}
$\varliminf_{\delta\rightarrow 0}\frac{\log N_{\delta}(F)}{-\log \delta}\leq \varliminf_{n\rightarrow \infty}\frac{\log N_{\delta_{n}}(F)}{-\log \delta_{n}}$, then
%\end{equation}
it suffices with applying a similar argument to the former.
%and now, we can continue using a similar argument to the \ref{eq:64} one.
%\item Finally, it suffices with applying the two previous results. Note that the existence of box-counting dimension (resp. fractal dimension I) implies that $\dim_{B}(F)=\overline{\dim}_{B}(F)=\underline{\dim}_{B}(F)$ (and the same argument is valid for fractal dimension I).
\end{enumerate}
\end{proof}
Since the sequence of diameters $\delta(\Gamma_{n})$ always decreases on a geometric way when working with self-similar sets (by its construction), then we can estimate the box-counting dimension of a self-similar set by means of its fractal dimension I, which is easier to calculate, since we are going to select its natural fractal structure. In this way, we present the next result.
%The condition we have described at expression \ref{eq:50} allows us to estimate the box-counting dimension of every self-similar set in terms of its fractal dimension I. Indeed, in every self-similar set we have that the sequence of diameters of the levels $\Gamma_{n}$ of the fractal structure $\ef$ decreases on a geometric way, by construction. So we obtain the following corollary.
\begin{cor}
Let $\ef$ be the natural fractal structure on a self-similar set $K$ provided with the euclidean distance. Then, the inequalities contained on the theorem \ref{teo:1} are verified, with $c$ being the maximum of the contraction factors $c_{i}$ associated with the contractions $f_{i}$ on the IFS $(X,\{f_{i}:i\in I\})$ whose associated self-similar set is $K$.
%For a self-similar set equipped with its natural fractal structure and the euclidean metric, the inequalities contained on the theorem \ref{teo:1} are verified, taking $c$ as the maximum of the contraction factors $c_{i}$ associated with the contractive mappings $f_{i}$ of the corresponding IFS $(X,\{f_{i}:i\in I\})$.
\end{cor}
%\begin{proof}
%Let $I$ be a finite index set, with $(K,\{f_{i}: i\in I\})$ being the IFS whose associated self-similar set is $K$.
%Note that using the natural fractal structure associated with this self-similar set $K$, the diameters of the elements on each level $\Gamma_{n}$ of the fractal structure $\ef$ decrease in a geometric way, so that there exists $c\in (0,1)$ such that $\delta(K,\Gamma_{n+1})\leq c\ \delta(K,\Gamma_{n})$ for all $n\in \mathbb{N}$, with $c$ being the maximum of $\{c_{i}:i\in I\}$, where $c_{i}$ are the contraction factors associated with the contractive mappings $f_{i}$ for all $i\in I$. Hence, it suffices with applying the theorem \ref{teo:1} in order to get the result.
%\end{proof}
%Accordingly, the fractal dimension I model we have studied in this section generalizes the box-counting one in the context of the euclidean spaces $\mathbb{R}^{d}$, and on the other hand, under certain conditions over the size of the elements on each level of the fractal structure $\ef$, we have found an upper bound of its box-counting dimension in terms of its fractal dimension I.
Fractal dimension I depends on the fractal structure $\ef$ we utilize on the space. Indeed, we show it in the following remark.

%Finally, we present the key fact that fractal dimension I depends on the fractal structure $\ef$ we are using on each case.
\begin{obs}\label{obs:2}
Let $X$ be a subspace of an euclidean space. Then, it is possible to obtain different values for its fractal dimension I, depending on the fractal structure $\ef$ we select in order to calculate it.
\end{obs}
\begin{proof}
First of all, let $\ef_{1}=\{\Gamma_{1,n}:n\in \mathbb{N}\}$ be a fractal structure whose levels are given by $\Gamma_{1,n}=\{[\frac{k}{2^{n}},\frac{k+1}{2^{n}}]:k\in \mathbb{Z}\}$ for all $n\in \mathbb{N}$, as the natural fractal structure on the real line induced on the middle third Cantor set $C$. Thus, by means of theorem \ref{prop:1}, we have that $\dim_{\ef_{1}}^{1}(C)=\dim_{B}(C)=\frac{\log 2}{\log 3}$ (see \cite[Example 3.3]{FAL90}).
%consider the GF-space $(C,\ef_{1})$ where $C$ is the middle third Cantor set, and $\ef_{1}$ is the fractal structure given as the countable family of coverings $\ef_{1}=\{\Gamma_{1,n}:n\in \mathbb{N}\}$, with $\Gamma_{1,n}=\{[\frac{k}{2^{n}},\frac{k+1}{2^{n}}]:k\in \mathbb{Z}\}$ for all natural number $n\in \mathbb{N}$. In other words, $\ef_{1}$ is the natural fractal structure on the real line induced on $C$. Thus, by theorem \ref{prop:1}, $\dim_{\ef_{1}}^{1}(C)=\dim_{B}(C)=\frac{\log 2}{\log 3}$ (see \cite[Example 3.3]{FAL90}).
%On the other hand, an application of the theorem \ref{teo:1} gives $\dim_{\ef_{1}}^{1}(C)=\dim_{B}(C)\leq \frac{\log 2}{\log 3}$, since $c_{i}=\frac{1}{3}$ for all $i\in \{1,2\}$.\newline
On the other hand, let $\Gamma_{2}$ be the natural fractal structure on $C$ as a self-similar set. Then, an easy calculation leads to $\dim_{\ef_{2}}^{1}(C)=1$, since on each level of the fractal structure $\ef_{2}$, there are $2^{n}$ subintervals of length equal to $\frac{1}{3^{n}}$.
\end{proof}
The theoretical justification of this result is based on the idea that the first version of fractal dimension regards all the elements on each level $\Gamma_{n}$ of the fractal structure $\ef$ as having the same ``size'', namely, $\frac{1}{2^{n}}$.

\section{Fractal Dimension on GF-spaces: a second version}

We have just presented a first method in order to calculate the fractal dimension of a subset. In this way, the use of GF-spaces leads to generalize the box-counting dimension, which only works over metrizable spaces (and in particular, over the euclidean ones), and it enables the use of a larger collection of fractal structures in order to calculate the fractal dimension than the box-counting method.
As seen on the previous section, the definition \ref{def:1} of fractal dimension I regards each element on the family $\mathcal{A}_{n}(F)$ as having the same ``size'', equal to $\delta=\frac{1}{2^{n}}$ for all $n\in \mathbb{N}$. So that, the natural fractal structure of $\delta-$cubes which we use on the euclidean spaces in order to determine the box-counting dimension of a given subset, can be extended to other kinds of tilings such as triangulations on the plane.
%Note that it results interesting because, for instance, all surface have a triangulation.

%Indeed, the definition \ref{def:1} of fractal dimension I, regards each element $A\in \Gamma_{n}$ of the fractal structure $\ef$ as having the same diameter, equal to $\frac{1}{2^{n}}$. Of course, this idea is inspired in the natural fractal structure we can define on each euclidean space $\mathbb{R}^{d}$.

Recall also that in remark \ref{obs:2} we have shown the fact that fractal dimension I depends on the fractal structure we use on each case. Taking it into account, we are going to propose another model in order to calculate the fractal dimension of any subset on a GF-space, but this time, we pretend that our definition can consider the possibility that the elements on each family $\mathcal{A}_{n}(F)$ could present diameters not necessarily equal to $\frac{1}{2^{n}}$.
%By the way, recall that the diameter of a subset $F\subset X$ on the level $\Gamma_{n}$ of the fractal structure $\ef$ has been defined as
%%\begin{equation}
%$\delta(F,\Gamma_{n})=\sup\{\delta(A): A\in \Gamma_{n}, A\cap F\neq \emptyset\}$.
%%\end{equation}

To this end, we will use the more general concept of distance function (see
\cite{ROM96}):
\begin{defn}
A mapping $d:X\times X\rightarrow \mathbb{R}^{+}$ defined on a set $X$ is said
to be a distance function, or merely a distance on $X$, if it verifies that
$d(x,x)=0$ for all $x\in X$.
\end{defn}

Diameters of subsets, coverings, etc. with respect to a distance function are
defined in the same way that for a metric.

The definition that follows is going to be called as the second version of the fractal dimension definition.
\begin{defn}\label{dimf:2}
Let $\ef$ be a fractal structure on a distance space $(X,\rho)$ and let $N_{n}(F)$ be the number of elements of $\mathcal{A}_{n}(F)$ for all $n\in \mathbb{N}$. Thus, the (lower/upper) fractal dimensions II of a non-empty bounded subset $F$ of $X$ respectively, are defined as the (lower/upper) limit:
%\begin{equation}
%\underline{\dim}^{2}_{\ef}(F)=\varliminf_{n\rightarrow \infty} \frac{\log N_{n}(F)}{-\log \delta(F,\Gamma_{n})}
%\end{equation}
%\begin{equation}
%\overline{\dim}^{2}_{\ef}(F)=\varlimsup_{n\rightarrow \infty} \frac{\log N_{n}(F)}{-\log \delta(F,\Gamma_{n})}
%\end{equation}
%and if those expressions agree, we will say that their common value is the fractal dimension II of $F$, namely
\begin{equation}
\dim_{\ef}^{2}(F)=\lim_{n\rightarrow \infty}\frac{\log N_{n}(F)}{-\log \delta(F,\Gamma_{n})}
\end{equation}
if this limit exists, where $\mathcal{A}_{n}(F)$ is defined at \ref{eq:83} for all $n\in \mathbb{N}$.
\end{defn}

%Recall that a fractal structure $\ef$ is said to be finite if all levels $\Gamma_{n}$ are finite coverings of the whole space $X$.
It is also important that working with finite fractal structures simplifies the necessary calculations in order to determine the fractal dimension of a subset.
In this way, the next remark shows that metrizable and separable topological spaces are useful for our purposes.

%Taking it into account, we affirm that, in particular, metrizable and separable topological spaces are going to be interesting for our purposes, as we can find in the next remark.
\begin{obs}
A metrizable space is second-countable if and only if it is separable (see \cite[Theorem 5.7]{SG02C}). On the other hand, it is known that every topological space which is second-countable has a compatible finite fractal structure (see \cite[Theorem 4.3]{SG02B}).
%Accordingly, we obtain the fact that metrizable and separable spaces have a finite fractal structure, so they are going to be very interesting for our study on fractal dimension.
\end{obs}

The second model we have presented in order to calculate the fractal dimension of a given subset enables the use of fractal structures whose elements could present different diameters on each level. In this way, we have used the quantity $\delta(F,\Gamma_{n})$ in order to do this. Nevertheless, it would be also possible to consider $\delta(\Gamma_{n})$ instead of $\delta(F,\Gamma_{n})$ but it implies certain disadvantages. This idea is included in the next example.
%Note that we have used the quantity $\delta(F,\Gamma_{n})$ in order to define the fractal dimension II of $F\subset X$, so that we have to consider the supremum of the diameters of the elements $A$ on each level $\Gamma_{n}$ of the fractal structure $\ef$ that meets $F$. However, we could have taken $\delta(\Gamma_{n})$ instead of $\delta(F,\Gamma_{n})$, but it implies certain disadvantages

%Indeed, let $\ef$ be a finite fractal structure on a distance space $(X,\rho)$ where $\rho$ is not a bounded distance function. Thus, $\delta(\Gamma_{n})=\infty$ for all $n\in \mathbb{N}$, which does not allow to calculate the fractal dimension II of any subset $F$.

%$(X,\ef)$ be a GF-space equipped with a not bounded distance function, being $\ef$ a finite fractal structure. Then, we get that $\delta(\Gamma_{n})=\infty$ for all $n\in \mathbb{N}$, so that we cannot calculate the fractal dimension II for any subset $F\subset X$. Indeed we have the following example.

\begin{ejem}
Let $\ef$ be a finite fractal structure whose levels are given by $\Gamma_{n}=\{[\frac{k_{1}}{2^{n}},\frac{k_{1}+1}{2^{n}}]\times \ldots \times [\frac{k_{d}}{2^{n}},\frac{k_{d}+1}{2^{n}}]:k_{i}\in \{-n2^{n},\ldots, n2^{n}-1\}, i\in \{1,\ldots, d\} \}\cup \{\mathbb{R}^{d}\setminus (-n,n)^{d}\}$ for all natural number $n$, defined on the euclidean space $\mathbb{R}^{d}$. Thus, for all subset $F$ of $\mathbb{R}^{d}$, we have that $\delta(\Gamma_{n})=\infty$ for all $n\in \mathbb{N}$, but nevertheless, there exists a natural number $n_{0}$ such that $\delta(F,\Gamma_{n})<\infty$ for all $n\geq n_{0}$.
%Hence, it results interesting the use of the quantity $\delta(F,\Gamma_{n})$ while calculating the fractal dimension II of $F$.

%Let $(X,\ef)$ be a GF-space, where $X=\mathbb{R}^{d}$ and the fractal structure $\ef$ is defined as the countable family of coverings $\ef=\{\Gamma_{n}:n\in \mathbb{N}\}$, with $\Gamma_{n}=\{[\frac{k_{1}}{2^{n}},\frac{k_{1}+1}{2^{n}}]\times \ldots \times [\frac{k_{d}}{2^{n}},\frac{k_{d}+1}{2^{n}}]:k_{1},\ldots, k_{d}\in \{-n2^{n},\ldots, n2^{n}-1\} \}\cup \{\mathbb{R}^{d}\setminus (-n,n)^{d}\}$ for all $n\in \mathbb{N}$. It is clear that $\ef$ is a finite fractal structure. One of our goals would consist of calculating its box-counting dimension and its fractal dimension II, and it would be desirable that such values agree. Nevertheless, given any bounded subset $F\subset X$, we have that $\delta(\Gamma_{n})=\infty$ for all $n\in \mathbb{N}$. However, there exists a natural number $n_{0}\in \mathbb{N}$ such that $\delta(F,\Gamma_{n})< \infty$ for all $n\geq n_{0}$. Hence, the calculation of the fractal dimension of any bounded subset of a given GF-space whose fractal structure $\ef$ is finite is a very interesting question when using the quantity $\delta(F,\Gamma_{n})$ on fractal dimension II definition.
\end{ejem}

Taking into account the definition \ref{dimf:2} as well as the \ref{def:1} one, it results clear that fractal dimensions I and II are going to agree if we select any fractal structure $\ef$ such that $\delta(F,\Gamma_n)=\frac{1}{2^{n}}$ for all $n\in \mathbb{N}$. In this way, we are looking for a general condition for the elements of the fractal structure $\ef$, in order to show that the two fractal dimension definitions are equal. First of all, we are going to pay attention to starbase fractal structures defined on a suitable class of distance spaces.
%Recall that a fractal structure $\ef$ is said to be starbase when $St(x,\ef)$ is a neighborhood base for every $x\in X$. Working with starbase fractal structures will allow us to show that fractal dimension II generalizes to fractal dimension I, and recall that we know a useful standard condition in order to get starbase GF-spaces (see proposition \ref{prop:7}). In order to demonstrate this result, we are going to introduce the concept of semimetric associated with a fractal structure $\ef$. First, we recall the definition of a semimetric over a topological space $X$ that can be found on \cite[Def. 9.5]{GRU84}:
In particular, a semimetric on a topological space $X$ (see \cite[Def. 9.5]{GRU84}) is a non-negative real mapping $d$ defined on $X \times X$, and verifying the following conditions: (i) $d(x,y)=0$ if and only if $x=y$, (ii) the mapping $d$ is symmetric, and (iii) The family $\{B_{d}(x,\varepsilon): \varepsilon> 0\}$ is a neighborhood base for all $x\in X$, namely, the topology induced by the semimetric $d$ agrees with the starting topology. \label{con:3}
%\begin{defn}\label{def:3}
%A mapping $d:X\times X\rightarrow \mathbb{R}^{+}$ defined on a topological space $(X,\tau)$ is called a semimetric if it verifies the following conditions:
%%Let $(X,\tau)$ be a topological space. A semimetric on $X$ is a mapping $d:X\times X\rightarrow \mathbb{R}^{+}$ verifying the following conditions:
%\begin{enumerate}
% \item $d(x,y)=0$ if and only if $x=y$.
% \item The mapping $d$ is symmetric.
% %namely $d(x,y)=d(y,x)$ for every $x,y\in X$.
% \item The family $\{B_{d}(x,\varepsilon): \varepsilon> 0\}$ is a neighborhood base for all $x\in X$, namely, the topology induced by the semimetric $d$ agrees with the starting topology: $\tau=\tau_{d}$. \label{con:3}
%\end{enumerate}
%\end{defn}
By the way, it is also possible to define a suitable semimetric on a GF-space. Note that starbase fractal structures enables to verify the condition (iii). The concept of semimetric associated with a fractal structure $\ef$ was introduced on \cite[Theorem 6.5]{SGnd}.
%Now, we recall the concept of semimetric associated with a fractal structure $\ef$ (see \cite[Theorem 6.5]{SGnd}):
%\begin{defn}
Indeed, let $(X,\ef)$ be a starbase GF-space. The semimetric associated with the fractal structure $\ef$ can be described as the mapping $\rho:X\times X\rightarrow \mathbb{R}^{+}$ given by:
\begin{equation}\label{eq:72}
\rho(x,y)=\left\{
\begin{array}{ll}
\hbox{$0$}\ \hfill \ldots & \hbox{$x=y$} \\
\hbox{$\frac{1}{2^{n}}$}\ \ldots & \hbox{$y\in St(x,\Gamma_{n})\setminus St(x,\Gamma_{n+1})$} \\
\hbox{$1$}\ \hfill \ldots & \hbox{$y\notin St(x,\Gamma_{1})$}
\end{array}
\right.
\end{equation}
%\end{defn}

%\begin{obs}
It follows from expression \ref{eq:72} that $B_{\rho}(x,\frac{1}{2^{n}})=St(x,\Gamma_{n+1})$ for all $n\in \mathbb{N}$ and all $x\in X$, and since $\ef$ is a starbase fractal structure, we conclude that the topology induced by the semimetric $\rho$ agrees with the topology induced by the fractal structure.
%Taking into account the previous definition we have that $B_{\rho}(x,\frac{1}{2^{n}})=St(x,\Gamma_{n+1})$ for all $n\in \mathbb{N}$, which is a neighborhood base for every $x\in X$ since $\ef$ is a starbase fractal structure. Hence, the topology induced by the semimetric $\rho$ associated with $\ef$, matches with the topology induced by the fractal structure $\ef$.
%\end{obs}
%Additionally, we need to ensure that all consecutive levels of the fractal structure $\ef$ are not the same thing. In order to do this, we are going to suppose that for all natural number $n$ there exists a point of $F$ such that its stars at consecutive levels $\Gamma_{n}$ and $\Gamma_{n+1}$ are not equal.
%Accordingly, we have the next theorem, where we show that fractal dimension II generalizes fractal dimension I on GF-spaces equipped with the semimetric associated with the fractal structure.
%Thus, we have almost all the necessary hypothesis in order to show that fractal dimension II generalizes to fractal dimension I while considering starbase GF-spaces. Moreover, we need to ensure the fact that all consecutive levels of the fractal structure $\ef$ are not the same thing. So that, we can include a condition consisting on the fact that for each natural number $n\in \mathbb{N}$, there exists some point $x\in F$ such that its stars at consecutive levels $\Gamma_{n}$ and $\Gamma_{n+1}$ are not equal.

\begin{teo}\label{prop:8}
Let $(X,\ef)$ be a starbase GF-space equipped with the semimetric associated with the fractal structure $\ef$, and let $F$ be a subset of $X$. Suppose also that for all $n\in \mathbb{N}$ there exists $x\in F$ such that $St(x,\Gamma_{n})\neq St(x,\Gamma_{n+1})$. Then,
\begin{enumerate}
\item $\overline{\dim}_{\ef}^{1}(F)=\overline{\dim}_{\ef}^{2}(F)$.
\item $\underline{\dim}_{\ef}^{1}(F)=\underline{\dim}_{\ef}^{2}(F)$.
\item If there exists any of the fractal dimensions I or II, then
$\dim_{\ef}^{1}(F)=\dim_{\ef}^{2}(F)$.
\end{enumerate}
\end{teo}

\begin{proof}
%\begin{enumerate}
%\item %Using the semimetric associated with the fractal structure $\ef$, we obtain that $\delta(A)\leq \frac{1}{2^{n}}$ for all $A\in \Gamma_{n}$.

Given $n \in \mathbb{N}$, there exists $x\in F$ such that $St(x,\Gamma_{n})\neq St(x,\Gamma_{n+1})$. This implies that there exists $A\in \Gamma_n$ with $x \in A$ and such that $A\not\subseteq St(x,\Gamma_{n+1})$. Thus, $\delta(A)=\frac{1}{2^{n}}$, and then, $\delta(F,\Gamma_{n})=\frac{1}{2^{n}}$. Therefore,
%Since for every natural number $n\in \mathbb{N}$ there exists $x\in F$ such that $St(x,\Gamma_{n})\neq St(x,\Gamma_{n+1})$, then we can affirm that there exists $A\in \Gamma_{n}$ such that $x\in A\not\subseteq St(x,\Gamma_{n+1})$. Hence, $\delta(A)=\frac{1}{2^{n}}$ for all $A\in \Gamma_{n}$, which implies that $\delta(F,\Gamma_{n})=\frac{1}{2^{n}}$.
%\begin{equation}\label{eq:5.1}
$\varlimsup_{n\rightarrow \infty}\frac{\log N_{n}(F)}{-\log \delta(F,\Gamma_{n})}=\varlimsup_{n\rightarrow \infty}\frac{\log N_{n}(F)}{n\log 2}$.
%\end{equation}
%so taking upper limits when $n\rightarrow \infty$ at expression \ref{eq:5.1}, we get the desired result.
%\item %It suffices with taking lower limits as $n\rightarrow \infty$ at expression \ref{eq:5.1}.
The case for lower limits may be dealt with in the same way.

%\item It is an immediate consequence of the two previous results.
%\end{enumerate}
\end{proof}

Fractal dimension II also generalizes the fractal dimension I as well as the box-counting dimension on the euclidean space $\mathbb{R}^{d}$ equipped with its natural fractal structure. In this way, we get a result like theorem \ref{prop:1} on the previous section.
%In this way, we get a similar result to theorem \ref{prop:1} with fractal dimension I.

\begin{teo}\label{prop:2}
Let $\ef$ be the natural fractal structure on the euclidean space $\mathbb{R}^{d}$, and let $F$ be a subset of $\mathbb{R}^{d}$. Then ${\dim}_{\ef}^{1}(F)={\dim}_{\ef}^{2}(F)=\dim_{B}(F)$ (the same is true for upper and lower dimensions).
%\begin{enumerate}
%\item $\overline{\dim}_{\ef}^{1}(F)=\overline{\dim}_{\ef}^{2}(F)=\overline{\dim}_{B}(F)$.
%\item $\underline{\dim}_{\ef}^{1}(F)=\underline{\dim}_{\ef}^{2}(F)=\underline{\dim}_{B}(F)$.
%\item If there exists any of the previous dimensions, then ${\dim}_{\ef}^{1}(F)={\dim}_{\ef}^{2}(F)=\dim_{B}(F)$.
%\end{enumerate}
\end{teo}

\begin{proof}
By theorem \ref{prop:1}, $\overline{\dim}_{\ef}^{1}(F)=\overline{\dim}_{B}(F)$.

%In order to calculate the box-counting dimension of $F$, let $N_{\delta}(F)$ be as (\ref{bc3}) on equivalent box-counting definition seen at preliminary section $2$.
%we are going to consider the quantity $N_{\delta}(F)$ as the number of $\delta-$cubes which meet $F$.
%\begin{enumerate}
%\item
%First of all, it is clear that $\overline{\dim}_{\ef}^{1}(F)=\overline{\dim}_{B}(F)$ by means of theorem \ref{prop:1}.

Note also that $\delta(A)=\frac{\sqrt{d}}{2^{n}}$ for all $A\in \Gamma_{n}$. Accordingly, we have that
%\begin{equation}
%\begin{split}
$\varlimsup_{n\rightarrow \infty} \frac{\log N_{n}(F)}{-\log \delta(F,\Gamma_{n})}
%=\varlimsup_{n\rightarrow \infty} \frac{\log N_{n}(F)}{-\log \sqrt{d}+n \ \log 2}
=\varlimsup_{n\rightarrow \infty} \frac{\log N_{n}(F)}{n \ \log 2}$
%\end{split}
%\end{equation}
which implies that $\overline{\dim}_{\ef}^{1}(F)=\overline{\dim}_{\ef}^{2}(F)$. The case for lower limits may be dealt with in the same way.
%\item
%Using an analogous argument to the previous one, we obtain that $\underline{\dim}_{\ef}^{2}(F)=\underline{\dim}_{\ef}^{1}(F)$.
%\item
%It suffices with applying the two previous results and taking into account the fact that the existence of any of those dimensions implies that its upper and lower dimensions are equal to itself.
%\end{enumerate}
\end{proof}

As seen on proposition \ref{prop:6} for fractal dimension I, we can also study some analytical properties for the fractal dimension
II definition. In this way, let $\ef$ be a fractal structure on a distance space $(X,\rho)$. First of all,
it is clear that both $\underline{\dim}^{2}_{\ef}$ and
$\overline{\dim}^{2}_{\ef}$ are monotonic. Furthermore, since fractal
dimension II generalizes fractal dimension I (by means of theorem \ref{prop:8}),
then any counterexample available for fractal dimension I remains valid for fractal
dimension II. In particular, we can use those given for statements
\ref{it:3}, \ref{it:4} and \ref{it:5} on proposition \ref{prop:6} for fractal
dimension II. Nevertheless, unlike fractal dimension I, we affirm that fractal
dimension II does not verify the finite stability, as shown in the next example.
\begin{ejem}
%Let $C_{1}$ be the middle third Cantor set on $[0,1]$ with $\ef_{1}$ being the natural fractal structure on $C_{1}$ as a self-similar set.
%Let $\ef_{1}$ be the
%natural fractal structure, as a self-similar set, on the middle-third Cantor set
%on $[0,1]$, which we
%denote by $F_{1}$.
Let $\ef_{1}$ be the natural fractal structure on $C_{1}$ as a self-similar set, where $C_{1}$ is the middle third Cantor set on $[0,1]$.
Let also $\ef_{2}$ be a fractal structure on $C_{2}=[2,3]$
given by $\ef_{2}=\{\Gamma_{2,n}:n\in \mathbb{N}\}$, with
$\Gamma_{2,n}=\{[\frac{k}{2^{2n}},\frac{k+1}{2^{2n}}]:k\in
\{2^{2n+1},2^{2n+1}+1,\ldots, 3\cdot 2^{2n}-1\}\}$ for all natural number $n$.
Consider now $\ef=\{\Gamma_{n}:n\in \mathbb{N}\}$ as
a fractal structure on $C=C_{1}\cup C_{2}$, where $\Gamma_{n}=\Gamma_{1,n}\cup
\Gamma_{2,n}$ for all $n\in \mathbb{N}$. A
simple calculation leads to $\dim_{\ef}^{2}(C_{1})=\frac{\log 2}{\log 3}$ and
$\dim_{\ef}^{2}(C_{2})=1$, while $\dim_{\ef}^{2}(C)=\frac{\log 4}{\log 3}> 1$.
\end{ejem}
Recall that on remark \ref{obs:2} we have shown that the fractal dimension I and the box-counting dimension of the middle third Cantor $C$ set did not agree: indeed, we got that $\dim_{B}(C)=\frac{\log 2}{\log 3}$, unlike $\dim_{\ef}^{1}(C)=1$. Note that these dimensions were calculated respect to different fractal structures: we used the natural fractal structure on the real line induced on $C$ in order to calculate the box-counting dimension, and on the other hand, we selected the natural fractal structure on the self-similar set to obtain the fractal dimension I of such space. However, the fractal dimension II agrees with the box-counting dimension of $C$. In this way, we consider again the natural fractal structure on $C$ as a self-similar set. Note that
%\begin{equation}
%\begin{split}
$\dim_{\ef}^{2}(C)
%&=\lim_{n\rightarrow \infty} \frac{\log N_{n}(C)}{-\log \delta(C,\Gamma_{n})}\\
=\lim_{n\rightarrow \infty} \frac{\log 2^{n}}{-\log 3^{-n}}
=\frac{\log 2}{\log 3}=\dim_{B}(C)$
%\end{split}
%\end{equation}
since on each level $\Gamma_{n}$ of the fractal structure there are $2^{n}$ elements with diameters equal to $3^{-n}$.

Nevertheless, though the value obtained from fractal dimension I of $C$ may seems counterintuitive at a first time, it is possible to explains it by means of the fractal dimension II. Indeed, the reason is that fractal dimension I only depends on the fractal structure we select in order to calculate it. We study this fact in the next remark.
%the fact that the value obtained for fractal dimension I calculation depends on the fractal structure $\ef$ we are using in the starting GF-space $(X,\ef)$. Specifically, we considered the GF-spaces $(C,\ef_{i})$ with $C$ being the middle third Cantor set. The fractal structure was not other thing that the natural fractal structure on the real line for the case $i=1$, and the natural fractal structure as a self-similar set, for the case $i=2$. Thus, fractal dimension I and box-counting dimension (considered respect to its corresponding fractal structures) did not agree.\newline
%However, if we consider the natural fractal structure on this strict self-similar set, then we have that on each level $\Gamma_{n}$ of the fractal structure $\ef$ there are $2^{n}$ elements with diameters equal to $\frac{1}{3^{n}}$. Thus, the effective calculation for fractal dimension II of $C$ returns us the following:
%\begin{equation}
%\begin{split}
%\dim_{\ef}^{2}(C)
%&=\lim_{n\rightarrow \infty} \frac{\log N_{n}(C)}{-\log \delta(C,\Gamma_{n})}\\
%&=\lim_{n\rightarrow \infty} \frac{\log 2^{n}}{-\log \frac{1}{3^{n}}}\\
%&=\frac{\log 2}{\log 3}=\dim_{B}(C).
%\end{split}
%\end{equation}
%Note that fractal dimension II and box-counting dimension calculated over the middle third Cantor set agree, although we are using different fractal structures on $C$.
\begin{obs} \label{obs:5}
Fractal dimension I only depends on the fractal structure, while fractal dimension II also depends on the diameter of the elements of each level of the fractal structure. We show this difference by constructing a family of spaces which from the fractal structure point of view are the same.
\end{obs}
\begin{proof}
Indeed, let us consider slight modifications on the construction of the middle third Cantor set, which we are going to denote by $C_{i}$, such that their associated contraction factors are $c_i\in [\frac{1}{3},\frac{1}{2})$ for the two similarities that give $C_i$. Thus, it is clear that $\delta(C_{i},\Gamma_{n})=c_i^{n}$ for all natural number $n$. Therefore, consider the natural fractal structure $\ef_{i}$ on each space $C_{i}$ as a self-similar set. Then, an easy calculation yields (or apply theorem \ref{bcfr})
%Furthermore, in order to calculate the fractal dimension II of $C_{i}$, we consider the family of GF-spaces $(C_{i},\ef_{i})$, where $\ef_{i}$ is the natural fractal structure as a strict self-similar set. Then,
%\begin{equation}
$\dim_{B}(C_{i})=\dim_{\ef}^{2}(C_{i})=\frac{\log 2}{-\log c_i}\rightarrow 1=\dim_{\ef}^{1}(C)$,
%\end{equation}
when $c_i\rightarrow \frac{1}{2}$.
\end{proof}

%By means of theorem \ref{teo:1}, we have established an upper bound for the box-counting dimension of any subset of a GF-space. In order to do this, we work under the hypothesis consisting of that the sequence of diameters of each level on the fractal structure decreases on a geometric way.
It is also possible to find an upper bound for the box-counting dimension of any subset in terms of its fractal dimension II, under the natural hypothesis $\delta(F,\Gamma_{n})\rightarrow 0$. Furthermore, we find a connection between Hausdorff dimension and fractal dimension II.

\begin{teo}\label{prop:3}
Let $\ef$ be a fractal structure on a metric space $(X,\rho)$ with $F$ being a subset of $X$, and suppose that $\delta(F,\Gamma_{n})\rightarrow 0$. Then,
\begin{enumerate}
\item $\dim_{H}(F)\leq \underline{dim}_{B}(F)\leq \underline{\dim}_{\ef}^{2}(F)$. \label{ap:1}
%Additionally, if $\dim_{\ef}^{2}(F)$ is  finite, then $\mathcal{H}^{s}(F)< \infty$.
\item If there exist both box-counting dimension and fractal dimension II of $F$, then $\dim_{B}(F)\leq \dim_{\ef}^{2}(F)$.
\item Suppose also that there exists a constant $c>0$ such that $\delta(F,\Gamma_{n})\leq c\ \delta(F, \Gamma_{n+1})$ for all natural number $n$. Then, \newline$\overline{\dim}_{B}(F)\leq \overline{\dim}_{\ef}^{2}(F)$.  \label{ap:2}
\end{enumerate}
\end{teo}

\begin{proof}
In order to calculate the box-counting dimension of $F$, let $N_{\delta}(F)$ be as ($5$) on equivalent box-counting definition seen at preliminary section $2$, and take $\delta_{n}=\delta(F,\Gamma_{n})$ for all $n\in \mathbb{N}$.
%Consider the decreasing sequence whose general term is $\delta_{n}=\delta(F,\Gamma_{n})$ which converges to $0$ when taking $n\rightarrow \infty$, and let $N_{\delta}(F)$ be the smallest number of sets of diameters at most $\delta$ that cover $F$, in order to calculate the box-counting dimension of $F$.
\begin{enumerate}
\item First of all, it is clear that $F\subset \cup\{A: A\in \mathcal{A}_{n}(F)\}$. Further, $\delta(A)\leq \delta(F,\Gamma_{n})=\delta_{n}$ for all $A\in \mathcal{A}_{n}(F)$. Accordingly, $F$ can be covered by $N_{n}(F)$ sets with diameters at most $\delta_{n}$, so applying [\cite{FAL90},Prop. 4.1] we get:
%\begin{equation}
%\begin{split}
$\dim_{H}(F)
\leq \underline{\dim}_{B}(F)
%&\leq \varliminf_{k\rightarrow \infty} \frac{\log n_{k}}{-\log \delta_{k}}\\
=\varliminf_{n\rightarrow \infty} \frac{\log N_{n}(F)}{-\log \delta_{n}}
=\underline{\dim}_{\ef}^{2}(F)$.
%\end{split}
%\end{equation}

%****************************LO SIGUIENTE ESTA MAL?
%On the other hand, if $\dim_{\ef}^{2}(F)<\infty$, then $N_{k}(F)<\infty$, so $n_{k}\cdot \delta_{k}^{s}=N_{k}(F)\cdot \delta(F,\Gamma_{k})^{s}\rightarrow 0$ when taking $k\rightarrow \infty$, since it is the product of a bounded value by the general term of a sequence that converges to $0$. Thus, an application of [\cite{FAL90}, Proposition 4.1] can allow us to assert that $\mathcal{H}^{s}(F)<\infty$, as desired.

\item It suffices with taking into account the previous result as well as the existence of both fractal dimensions.

\item Let $c\in (0,1]$ be such that $\delta(F,\Gamma_{n})\leq c\
\delta(F,\Gamma_{n+1})$ for all $n\in \mathbb{N}$. Thus, $\delta(F,\Gamma_{n})$
does not converges to $0$, which is a contradiction, so that $c>1$. Hence, there
exists $d\in (0,1)$ such that $\delta_{n+1}\geq d\ \delta_{n}$ for all $n\in
\mathbb{N}$. Therefore, using remark \ref{rem:F},
%\begin{equation}
%\begin{split}
$\overline{\dim}_{B}(F)
=\varlimsup_{n\rightarrow \infty}\frac{\log N_{\delta_{n}}(F)}{-\log \delta_{n}}
\leq \varlimsup_{n\rightarrow \infty}\frac{\log N_{n}(F)}{-\log \delta(F,\Gamma_{n})}$
%\end{split}
%\end{equation}
since $\delta(A)\leq \delta_{n}$ for all $A\in \mathcal{A}_{n}(F)$.
\end{enumerate}
\end{proof}

In this way, as a consequence of theorem \ref{prop:3}, it is possible to get an approximation of the box-counting dimension of a self-similar set in terms of its fractal dimension II.

\begin{cor} \label{cor:8}
Let $I=\{1,\ldots, m\}$ be a finite index set with $(\mathbb{R}^{d},\{f_{i}:i\in I\})$ being an IFS whose associated self-similar set is $K$. Let also $\ef$ be the natural fractal structure on the self-similar set, and let $F$ be a subset of $K$. Then,
\begin{enumerate}
\item $\underline{\dim}_{B}(F)\leq \underline{\dim}_{\ef}^{2}(F)$. \label{cor:81}
\item If there exist both box-counting dimension and fractal dimension II of $F$, then $\dim_{B}(F)\leq \dim_{\ef}^{2}(F)$. \label{cor:82}
\item Suppose that there exists $i\in I$ such that $f_{i}$ is a bilipschitz contractive mapping. Then $\overline{\dim}_{B}(F)\leq \overline{\dim}_{\ef}^{2}(F)$. In particular, this inequality is verified for strict self-similar sets. \label{cor:83}
\end{enumerate}
\end{cor}

\begin{proof}
In order to calculate the box-counting dimension of $F$, let $N_{\delta}(F)$ be as ($5$) on equivalent box-counting definition seen at preliminary section $2$.
\begin{enumerate}
\item Since $K$ is a self-similar set, it is clear that $\delta(F, \Gamma_{n})\rightarrow 0$ for all $F\subset K$, so it suffices with applying the theorem \ref{prop:3}.\ref{ap:1}.
%we have that $\underline{\dim}_{B}(F)\leq \underline{\dim}_{\ef}^{2}(F)$.
\item This is by the first item.
%the box-counting dimension and the fractal dimension II of $F\subset K$, then $\dim_{B}(F)=\underline{\dim}_{B}(F)\leq \underline{\dim}_{\ef}^{2}(F)=\dim_{\ef}^{2}(F)$.
\item Let $f_{i}$ be a bilipschitz contraction. Thus, there exist constants $L_{i}$ and $c_{i}$ with $0<L_{i}<c_{i}<1$, such that $L_{i}\ d(x,y)\leq d(f_{i}(x),f_{i}(y))\leq c_{i}\ d(x,y)$ for all $x,y\in K$. Now, consider $A\in \mathcal{A}_{n}(F)$ verifying that $\delta(A)=\delta(F,\Gamma_{n})$. By definition of supremum, we have that for all $\varepsilon>0$, there exist $x,y\in A$ such that $d(x,y)> \delta(A)-\varepsilon$. Let $B=f_{i}(A)\in \Gamma_{n+1}$. Thus, we have that $d(f_{i}(x),f_{i}(y))\geq L_{i}\ d(x,y)> L_{i}\ (\delta(A)-\varepsilon)$,  which implies that $\delta(B)=\delta(f_{i}(A))\geq L_{i}\ \delta(A)$. Accordingly, $\delta(F,\Gamma_{n+1})\geq \delta(B)\geq L_{i}\ \delta(F,\Gamma_{n})$, so that it can be applied the proposition \ref{prop:3}.\ref{ap:2} in order to get the result.
\end{enumerate}
\end{proof}

In the next result, we look for properties on the natural fractal structure of an euclidean space in order to generalize theorem \ref{prop:2}.

% Consider the GF-space $(\mathbb{R}^{d},\ef)$ equipped with the euclidean distance, where $\ef$ is the fractal structure associated in a natural way to $\mathbb{R}^{d}$, composed by $\frac{1}{2^{n}}-$cubes on each level $\Gamma_{n}$. Taking it into account, we recall an interesting property of $\mathbb{R}^{d}$, consisting of that given any size $\delta>0$, it is verified that any subset $F\subset \mathbb{R}^{d}$ with $\delta(F)\leq \delta$ meets at most to $3^{d}$ $\delta-$cubes. Since box-counting dimension can be defined over each metrizable space, a similar property to the previous one on a more general context would allow us to find a suitable relationship between fractal dimension II and box-counting dimension. Indeed, we present the following result:
\begin{teo}\label{teo:9}
Let $\ef$ be a fractal structure on a metric space $(X,\rho)$ with $F$ being a subset of $X$, and suppose that there exists a natural number $k\in \mathbb{N}$ such that for all $n\in \mathbb{N}$, every subset $A$ of $X$ with $\delta(A)\leq \delta(F,\Gamma_{n})$, meets at most at $k$ elements of the level $\Gamma_{n}$ on the fractal structure $\ef$. Suppose also that $\delta(F,\Gamma_{n})\rightarrow 0$. Then,
\begin{enumerate}
\item $\underline{\dim}_{B}(F)\leq \underline{\dim}_{\ef}^{2}(F)\leq \overline{\dim}_{\ef}^{2}(F)\leq \overline{\dim}_{B}(F)$. Moreover, if there exists  $\dim_{B}(F)$, then $\dim_{B}(F)=\dim_{\ef}^{2}(F)$.
\item If there exists $c\in (0,1)$ such that $\delta(F,\Gamma_{n+1})\geq c\ \delta(F,\Gamma_{n})$, then $\overline{\dim}_{B}(F)=\overline{\dim}_{\ef}^{2}(F)$ and $\underline{\dim}_{B}(F)=\underline{\dim}_{\ef}^{2}(F)$.
\end{enumerate}
\end{teo}

\begin{proof}
In order to calculate the box-counting dimension of $F$, let $N_{\delta}(F)$ be as ($5$) on equivalent box-counting definition seen at preliminary section $2$, and let $\delta_{n}=\delta(F,\Gamma_{n})$ for all $n\in \mathbb{N}$.
%Consider the decreasing sequence whose general term is $\delta_{n}=\delta(F,\Gamma_{n})$, and let $N_{\delta}(F)$ be the smallest number of sets of diameter at most $\delta$ that cover $F$ in order to calculate the box-counting dimension of $F$.
\begin{enumerate}
\item First of all, note that by proposition \ref{prop:3} we have that $\underline{\dim}_{B}(F)\leq \underline{\dim}_{\ef}^{2}(F)$. On the other hand, the main hypothesis
%    consisting of the existence of a natural number $k$ such that for all $n\in \mathbb{N}$ it is verified that the set
%    $\{B\in \Gamma_{n}: A\cap B\neq \emptyset\}$
%    has at most $k$ elements, for all subset $A\subset X$ with $\delta(A)\leq \delta_{n}$,
    implies that $N_{n}(F)\leq k\ N_{\delta_{n}}(F)$ for all $n\in \mathbb{N}$, so it is clear that
    %\begin{equation}\label{eq:62}
%    \frac{\log N_{n}(F)}{-\log \delta_{n}}\leq \frac{\log k\ N_{\delta_{n}}(F)}{-\log \delta_{n}}
%    \end{equation}
%Now, it suffices with taking upper limits as $n\rightarrow \infty$ at expression \ref{eq:62} in order to get that follows:
%\begin{equation}
%\begin{split}\label{eq:63}
$\varlimsup_{n\rightarrow \infty}\frac{\log N_{n}(F)}{-\log \delta_{n}}
%&\leq \varlimsup_{n\rightarrow \infty}\frac{\log k+\log N_{\delta_{n}}(F)}{-\log \delta_{n}}\\
\leq\varlimsup_{n\rightarrow \infty}\frac{\log N_{\delta_{n}}(F)}{-\log \delta_{n}}
\leq \overline{\dim}_{B}(F)$.
%\end{split}
%\end{equation}

%Accordingly, we have that $\overline{\dim}_{\ef}^{2}(F)\leq \overline{\dim}_{B}(F)$.

\item Note that, by remark \ref{rem:F},
%\begin{equation}
$\overline{\dim}_{B}(F)=\varlimsup_{n\rightarrow \infty}\frac{\log
N_{\delta_{n}}(F)}{-\log \delta_{n}}$, which
%\end{equation}
implies $\overline{\dim}_{B}(F)\leq \overline{\dim}_{\ef}^{2}(F)$. Now, an application of the first item leads to the opposite inequality.
%The opposite inequality may be dealt with in the same way as \ref{eq:63}.
The case for lower limits may be dealt with in the same way.
\end{enumerate}
\end{proof}

The condition used on theorem \ref{teo:9} in order to get the equality between fractal dimension II and box-counting dimension is necessary, as the next remark shows:

\begin{obs} \label{obs:6}
There exists a self-similar set $K$ on the euclidean plane provided with its natural fractal structure $\ef$, for which
%it is not verified the condition consisting of the existence of a natural number $k$ such that for all $n\in \mathbb{N}$, any subset $A$ of $\mathbb{R}^{2}$ verifying that $\delta(A)\leq \delta(K,\Gamma_{n})$ meets at most to $3^{2}$ elements of the level $\Gamma_{n}$ of the fractal structure $\ef$, and such that
$\dim_{B}(K)\neq \dim_{\ef}^{2}(K)$.
\end{obs}

\begin{proof}
%We affirm that such a hypothesis is necessary, so that the result we have shown is not improvable in this sense.
Indeed, let $I=\{1,\ldots, 8\}$ be a finite index set and let $(K,\{f_{i}:i\in I\})$ be an IFS whose associated strict self-similar set is the unit square on the euclidean plane. Consider the contraction mappings $f_{i}:\mathbb{R}^{2}\rightarrow \mathbb{R}^{2}$ given as follows:
\begin{equation}
f_{i}(x,y)=\left\{
\begin{array}{ll}
\hbox{$(\frac{x}{2},\frac{y}{4})+(0,\frac{i-1}{4})$}\ \ldots & \hbox{$i=1,2,3,4$} \\
\hbox{$(\frac{x}{2},\frac{y}{4})+(\frac{1}{2},\frac{i-5}{4})$}\ \ldots & \hbox{$i=5,6,7,8$}
\end{array}
\right.
\end{equation}
Let $\ef$ be the natural fractal structure on $K$ as a self-similar set.\newline
First of all, it is clear that $K$ is a self-similar set which is not a strict one. Further, we have that the contractive mappings $f_{i}$ have the same contraction factors, namely $c_{i}=\frac{1}{2}$ for all $i\in I$. It is also immediate the fact that $\dim_{B}(K)=2$. On the one hand, note that there are $8^{n}$ rectangles on each level $\Gamma_{n}$ of the fractal structure $\ef$, whose dimensions are $\frac{1}{2^{n}}\times \frac{1}{2^{2n}}$. Thus, it is verified that $\delta(A)=\delta(K,\Gamma_{n})=\sqrt{\frac{1+2^{2n}}{2^{4n}}}$ for all $A\in \Gamma_{n}$. Hence, we can calculate the fractal dimension II of $K$ as follows:
%\begin{equation}
%\begin{split}
$\lim_{n\rightarrow \infty}\frac{\log N_{n}(K)}{-\log \delta(K,\Gamma_{n})}
=\lim_{n\rightarrow \infty}\frac{3n \log 2}{-\frac{1}{2}\log (\frac{1+2^{2n}}{2^{4n}})}
%&=\lim_{n\rightarrow \infty}\frac{3n \log 2}{-\frac{1}{2}\{\log (1+2^{2n})-\log 2^{4n}\}}\\
%&=\lim_{n\rightarrow \infty}\frac{3n \log 2}{-\frac{1}{2}\{2n \log 2-4n \log 2\}}\\
=\lim_{n\rightarrow \infty}\frac{3n \log 2}{n \log 2}
=3$.
%\end{split}
%\end{equation}
%On the other hand, we can find the relationship between $\delta(K,\Gamma_{n})$ and each side of the $\frac{1}{2^{n}}\times \frac{1}{2^{2n}}$-rectangles for all natural number $n$. Thus, we obtain that $\frac{\sqrt{\frac{1+2^{2n}}{2^{4n}}}}{\frac{1}{2^{2n}}}=\sqrt{1+2^{2n}}>2^{n}$, as well as $\frac{\sqrt{\frac{1+2^{2n}}{2^{4n}}}}{\frac{1}{2^{n}}}=\sqrt{1+\frac{1}{2^{2n}}}\geq \frac{1}{2^{n}}$ for all $n\in \mathbb{N}$. Therefore, each subset $A\subset K$ with diameter at most $\sqrt{\frac{2^{2n}+1}{2^{4n}}}$ is going to meet at most to $3\cdot 2^{n+1}$ elements $A\in \Gamma_{n}$ for all $n\in \mathbb{N}$. Accordingly, since this quantity depends on each natural number $n$, our main hypothesis on theorem \ref{teo:9} is not verified in this counterexample.
\end{proof}

%So far, we have seen that fractal dimensions I and II agree on the GF-space $\mathbb{R}^{d}$ equipped with the euclidean distance and its natural fractal structure $\ef$ (see Corollary \ref{prop:2}). In this way, it is an interesting question to determine some suitable conditions over the elements of a fractal structure $\ef$ in order to obtain the equality between fractal dimension methods I and II. In order to show a result of this kind, we recall when two positive-real-valued functions of natural variable have the same order:

%\begin{defn}
Let $f,g:\mathbb{N}\rightarrow \mathbb{R}^{+}$ be two sequences of positive real numbers. We say that $\mathcal{O}(f)=\mathcal{O}(g)$ iff they verify that
%\begin{equation}
$0<\lim_{n\rightarrow \infty}\frac{f(n)}{g(n)}<\infty$.
%\end{equation}
%\end{defn}
%The following lemma will be useful for our purposes:
%\begin{lema}\label{prop:4}
%Let $f,g:\mathbb{N}\rightarrow \mathbb{R}^{+}$ be two sequences of positive real numbers such that $\mathcal{O}(f)=\mathcal{O}(g)$, and suppose that there exists $\lim_{n\rightarrow \infty}\frac{h(n)}{f(n)}\in (0,\infty)$ with $h:\mathbb{N}\rightarrow \mathbb{R}^{+}$. Then there exists a constant $k\in (0,\infty)$ such that
%\begin{equation}
%\lim_{n\rightarrow \infty}\frac{h(n)}{f(n)}=k\cdot \lim_{n\rightarrow \infty}\frac{h(n)}{g(n)}
%\end{equation}
%\end{lema}
%
%\begin{proof}
%Indeed, the following is clear:
%\begin{equation}
%\begin{split}\label{eq:65}
%\lim_{n\rightarrow \infty}\frac{h(n)}{f(n)}\cdot \lim_{n\rightarrow \infty}\frac{f(n)}{g(n)}
%&=\lim_{n\rightarrow \infty} \Bigg(\frac{h(n)}{f(n)}\cdot \frac{f(n)}{g(n)}\Bigg)\\
%&=\lim_{n\rightarrow \infty}\frac{h(n)}{g(n)}
%\end{split}
%\end{equation}
%Note that $\lim_{n\rightarrow \infty}\frac{h(n)}{g(n)}\in (0,\infty)$, since $\mathcal{O}(f)=\mathcal{O}(g)$. Thus,
%\begin{equation}
%\lim_{n\rightarrow \infty}\frac{h(n)}{f(n)}=\lim_{n\rightarrow \infty}\frac{g(n)}{f(n)}\cdot \lim_{n\rightarrow \infty}\frac{h(n)}{g(n)}
%\end{equation}
%so it suffices with taking $k:=\lim_{n\rightarrow \infty}\frac{g(n)}{f(n)}\in (0,\infty)$.
%\end{proof}
In order to show that fractal dimensions I and II agree on a GF-space, we will need that all the elements on $\Gamma_{n}$ have a diameter of order $\frac{1}{2^{n}}$, as the next theorem establishes. %Indeed, if we think about the way on which we have built the fractal dimension I formula, it results natural to restrict our attention to those fractal structures $\ef$ whose elements have a diameter of this kind.
Its proof is left to the reader.
\begin{teo}\label{teo:8}
Let $\ef$ be a fractal structure on a metric space $(X,\rho)$. Let $F$ be a subset of $X$, and suppose that $\delta(A)=\delta(F,\Gamma_{n})$ for all $A\in \mathcal{A}_{n}(F)$ and $\mathcal{O}(\delta(F,\Gamma_{n}))=\mathcal{O}(\frac{1}{2^{n}})$ for all $n\in \mathbb{N}$. Then $\dim_{\ef}^{1}(F)=\dim_{\ef}^{2}(F)$ (the same is true for upper and lower dimensions).
%\begin{enumerate}
%\item $\overline{\dim}_{\ef}^{1}(F)=\overline{\dim}_{\ef}^{2}(F)$. \label{eq:51}
%\item $\underline{\dim}_{\ef}^{1}(F)=\underline{\dim}_{\ef}^{2}(F)$. \label{eq:52}
%\item If there exists any of the previous fractal dimensions, then \newline $\dim_{\ef}^{1}(F)=\dim_{\ef}^{2}(F)$. \label{eq:53}
%\end{enumerate}
\end{teo}

Recall that we have been able to calculate an upper bound for the box-counting dimension of every self-similar set equipped with its natural fractal structure in terms of its fractal dimension II.
However, it is possible to reach the equality between these two quantities under certain conditions on the self-similar set structure. In particular, we will get the result if the elements on each level of $\ef$ do not overlap too much, and because of the shape that the elements of the natural fractal structure on a self-similar set have, this restriction is going to be associated with the contractions $f_{i}$ of the corresponding IFS. Indeed, this property is the so called open set condition (see \cite{FAL90} and \cite{SCH94}):

\begin{defn}
Let $I=\{1,\ldots, m\}$ be a finite index set with $(X,\{f_{i}:i\in I\})$ being an IFS whose associated self-similar set is $K$. The contractions $f_{i}$ satisfy the open set condition iff there exists a non-empty bounded open set $V$ of $X$ such that $\bigcup_{i\in I}f_{i}(V)\subset V$, with $f_{i}(V)\cap f_{j}(V)=\emptyset$ for all $i\neq j$. Furthermore, if $V\cap K\neq \emptyset$, then the contractions $f_{i}$ are said to verify the strong open set condition.
\end{defn}

\begin{obs}\label{obs:4}
Open set condition and strong open set condition are equivalent in any euclidean space (see \cite{SCH94}).
\end{obs}

%\begin{obs}
%We have that $\overline{A}=\bigcap_{\varepsilon> 0}B(A,\varepsilon)$, for all $A\subset X$, with $(X,d)$ being a generic metric space. In fact, $x\in \bigcap_{\varepsilon> 0}B(A,\varepsilon)$ if and only if $d(x,A)< \varepsilon$ for all $\varepsilon>0$, which is equivalent to the condition $d(x,A)=0$, namely, $x\in \overline{A}$.\hfill $\cqd$
%\end{obs}

%The following result is a technical lemma which is necessary in order to reach the equality between fractal dimension II and box-counting dimension for a self-similar set (on an euclidean space) provided with its natural fractal structure:
\begin{lema}\label{lema1}
Let $I=\{1,\ldots, m\}$ be a finite index set and let $(\mathbb{R}^{d},\{f_{i}:i\in I\})$ be an IFS whose associated self-similar set is $K$. Suppose that $f_{i}$ are injective contractions verifying the open set condition for all $i\in I$. Then, there exists $\varepsilon>0$ and $x\in K$ such that for all natural number $n\in \mathbb{N}$ and all $\omega,u\in I^{n}$ with $\omega\neq u$, we have that $f_{\omega}(B(x,\varepsilon))\cap f_{u}(B(x,\varepsilon))=\emptyset$, where $B(x,\varepsilon)$ denotes the ball centered at $x\in X$ with radius $\varepsilon>0$, which is embedded in $\mathbb{R}^{d}$.
\end{lema}

\begin{proof}
Since $K\subset \mathbb{R}^{d}$ is a self-similar set, then we have that open set condition is equivalent to strong open set condition by remark \ref{obs:4}. Thus, there exists a non-empty bounded open set $O\subset \mathbb{R}^{d}$ such that \newline $\cup_{i\in I}f_{i}(O)\subset O$, with $f_{i}(O)\cap f_{j}(O)=\emptyset$ for all $i,j\in I$ such that $i\neq j$. Moreover, $O\cap K\neq \emptyset$, so we can take $x\in O\cap K\subset O$. Since $O$ is an open set, there exists $\varepsilon> 0$ such that $B=B(x,\varepsilon)\subset O$. We also know that $f_{i}(B)\subset O$ for all $i\in I$. The result will be shown by induction over the lenght of the words on $I^{n}$. In fact, let $\omega_{n}=i_{n}i_{n-1}\ldots i_{1}\in I^{n}$ and denote $f_{\omega_{n}}(F)=f_{i_{n}}\circ f_{i_{n-1}}\circ \ldots \circ f_{i_{1}}(F)$ for all $F\subset \mathbb{R}^{d}$. First of all, take $n=1$. Then we can select $\omega_{1}=i_{1}, u_{1}=j_{1}\in I$. Accordingly, we have that $f_{\omega_{1}}(B)\cap f_{u_{1}}(B)\subset f_{\omega_{1}}(O)\cap f_{u_{1}}(O)=\emptyset$, since the $f_{i}$ contractions verify the open set condition for all $i\in I$.
Suppose now that $f_{\omega_{n}}(B)\cap f_{u_{n}}(B)=\emptyset$, with $\omega_{n},u_{n}\in I^{n}$ such that $\omega_{n}\neq u_{n}$. Then, we distinguish the two following cases:
\begin{enumerate}
\item Suppose that $i_{n+1}=j_{n+1}$, so that $\omega_{n+1}=i_{n+1}\ i_{n}\ \ldots \ i_{1}$ and $u_{n+1}=i_{n+1}\ j_{n}\ \ldots \ j_{1}$, with $\omega_{n+1},u_{n+1}\in I^{n+1}$. Thus, using the injectivity of $f_{i_{n+1}}$ and the induction hypothesis, we conclude that $f_{\omega_{n+1}}(B)\cap f_{u_{n+1}}(B)=f_{i_{n+1}}(f_{\omega_{n}}(B))\cap f_{i_{n+1}}(f_{u_{n}}(B))=\emptyset$.

\item Suppose that $\omega_{n+1}\neq u_{n+1}$, so that $\omega_{n+1}=i_{n+1}\ i_{n}\ \ldots \ i_{1}$ and $u_{n+1}=j_{n+1}\ j_{n}\ \ldots \ j_{1}$, with $i_{n+1}\neq j_{n+1}$. Then, it is clear the following chain of inclusions:
%\begin{equation}\label{eq:65}
%\begin{split}
$f_{w_{n+1}}(B)\cap f_{u_{n+1}}(B)=f_{i_{n+1}}(f_{\omega_{n}}(B))\cap f_{j_{n+1}}(f_{u_{n}}(B))\subset f_{i_{n+1}}(O)\cap f_{j_{n+1}}(O)=\emptyset$.
%\end{split}
%\end{equation}
Note that $f_{\omega_{n}}(B)\subset O$ for all $\omega_{n}\in I^{n}$. Indeed, it is clear for words whose length is $1$, since it is obvious that $f_{i_{1}}(B)\subset O$. Suppose now that $f_{\omega_{n}}(B)\subset O$ for all $\omega_{n}\in I^{n}$. Then, it is verified that $f_{\omega_{n+1}}(B)=f_{i_{n+1}}(f_{\omega_{n}}(B))\subset f_{i_{n+1}}(O)\subset O$.
\end{enumerate}
\end{proof}

On \cite[Theorem 9.3]{FAL90} we find an interesting result which allows to calculate the box-counting dimension of a certain class of self-similar sets on the euclidean space $\mathbb{R}^{d}$ as the solution of a simple equation which only involves the contraction factors associated with each mapping $f_{i}$ of the corresponding IFS. Indeed, under the open set condition hypothesis, box-counting dimension agrees with Hausdorff dimension of such self-similar sets, and moreover, this value can be easily calculated from the mentioned expression. In this way, it would be an interesting result to reach the equality between box-counting dimension and fractal dimension II of a self-similar set whose contractions $f_{i}$ verify the open set condition. Moreover, the calculus of such quantity would become almost immediate from the number of contractive mappings of the IFS and its contraction factors.

Indeed, by means of the previous lemma, we present now the expected theorem.

\begin{teo}\label{bcfr}
Let $I=\{1,\ldots, m\}$ be a finite index set and let $(\mathbb{R}^{d},\{f_{i}:i\in I\})$ be an IFS whose associated strict self-similar set is $K$. Suppose that the similarities $f_{i}$ verify the open set condition and have equal similarity factors $c\in (0,1)$. Let $\ef$ be the natural fractal structure on $K$ as a self-similar set. Then,
\begin{equation}\label{eq:73}
\dim_{B}(K)=\dim_{\ef}^{2}(K)=\frac{-\log m}{\log c}
\end{equation}
\end{teo}

\begin{proof}
In order to calculate the box-counting dimension of $K$, let $N_{\delta}(K)$ be as ($6$) on the equivalent box-counting definition seen at preliminary section $2$, and let $\delta_{n}=\delta(K,\Gamma_{n})$ for all $n\in \mathbb{N}$.
First of all, corollary \ref{cor:8}.\ref{cor:83} implies that $\overline{\dim}_{B}(K)\leq \overline{\dim}_{\ef}^{2}(K)$. Next, we are going to show the opposite inequality. Note that $\delta_{n}=c^{n}\ \delta(K)$ for all natural number $n$, since $K$ is a strict self-similar set.
Applying lemma \ref{lema1}, there are so many disjoint balls with radius $\varepsilon_{n}=c^{n}\ \varepsilon$ with $\varepsilon> 0$, and centered in $K$, as the number of elements of $I^{n}$. Now, since $N_{\varepsilon_{n}}(K)$ is the largest number of such balls, it is obvious that the number of elements of $I^{n}$ is at most $N_{\varepsilon_{n}}(K)$, namely,
%\begin{equation}
$N_{n}(K)\leq N_{\varepsilon_{n}}(K) \label{f:1}$.
%\end{equation}
On the other hand, there exists $k>0$ such that
%\begin{equation}
$\delta(K,\Gamma_{n})= k\ \varepsilon_{n} \label{f:2}$.
%\end{equation}
Indeed, it suffices with taking $k=\frac{\delta(K)}{\varepsilon}$.
Therefore, it results clear that
%\begin{equation}
%\frac{\log N_{n}(K)}{-\log \delta(K, \Gamma_{n})}\leq \frac{\log N_{\varepsilon_{n}}(K)}{-\log k\ \varepsilon_{n}} \label{f:3}
%\end{equation}
%so taking upper limits when $k\rightarrow \infty$ at expression \ref{f:3}, we have that
%\begin{equation}
%\begin{split}
$\varlimsup_{n\rightarrow \infty}\frac{\log N_{n}(K)}{-\log \delta(K, \Gamma_{n})}
\leq \varlimsup_{n\rightarrow \infty}\frac{\log N_{\varepsilon_{n}}(K)}{-\log k\ \varepsilon_{n}}
=\varlimsup_{n\rightarrow \infty}\frac{\log N_{\varepsilon_{n}}(K)}{-\log \varepsilon_{n}}
\leq \overline{\dim}_{B}(K)$.
%\end{split}
%\end{equation}
%namely, $\overline{\dim}_{\ef}^{2}(K)\leq \overline{\dim}_{B}(K)$, as desired.
Now, we get the next chain of inequalities:
%\begin{equation}
%\begin{split}
$\underline{\dim}_{B}(K)\leq \underline{\dim}_{\ef}^{2}(K)\leq \overline{\dim}_{\ef}^{2}(K) \leq \overline{\dim}_{B}(K)$,
%\end{split}
%\end{equation}
where the first inequality is by corollary \ref{cor:8}.
Now, the existence of the box-counting dimension of $K$ implies the existence of the fractal dimension II of $K$ and the expected equality $\dim_{B}(K)=\dim_{\ef}^{2}(K)$.
Furthermore, apply \cite[Theorem 9.3]{FAL90} in order to get the last equality on \ref{eq:73}.
%In order to obtain the last equality, namely, $\dim_{\ef}^{2}(K)=\frac{-\log m}{\log c}$, it suffices with applying \cite[Theorem 9.3]{FAL90}.
\end{proof}

The hypothesis based on the equality of the contraction factors on theorem \ref{bcfr} is necessary.
By the example in remark \ref{obs:6}, the contractions have to be similarities, while by the next example, all contraction factors must be the same.

\begin{obs}\label{eq:71}
There exists a strict self-similar set $K$, whose similarities $f_{i}$ verify the open set condition and have different contraction factors $c_{i}$, such that $\dim_{B}(K)< \dim_{\ef}^{2}(K)$.
\end{obs}

\begin{proof}
Indeed, we can take the IFS $(\mathbb{R},\{f_{i}:i\in \{1,2\}\})$, whose associated self-similar set is $K$, where the similarities $f_{1}, f_{2}:\mathbb{R}\rightarrow \mathbb{R}$ are given by $f_{1}:x\mapsto \frac{x}{2}$, and $f_{2}:x\mapsto \frac{x+3}{4}$. It is clear that their associated contraction factors are $c_{1}=\frac{1}{2}$ and $c_{2}=\frac{1}{4}$ respectively, and it is also obvious the fact that $K$ is a strict self-similar set. We can also justify that the similarities $f_{i}$ satisfy the open set condition on an easy way: it suffices with taking $V=(0,1)$ as a suitable open set. Now, let $\ef$ be the natural fractal structure on $K$ as a self-similar set. Then, by \cite[Theorem 9.3]{FAL90}, we can calculate the box-counting dimension of $K$ as the solution of the equation $\frac{1}{2^{s}}+\frac{1}{4^{s}}=1$. Thus, we have that $\dim_{B}(K)=\frac{\log (\frac{1+\sqrt{5}}{2})}{\log 2}$.
On the other hand, it is clear that there are $2^{n}$ subintervals of $[0,1]$ on each level $\Gamma_{n}$ of the fractal structure $\ef$, where the largest of them has diameter equal to $\frac{1}{2^{n}}$ for all natural number $n\in \mathbb{N}$, which implies that $\dim_{\ef}^{2}(K)=1>\dim_{B}(K)$.
\end{proof}
% The underlying result on remark \ref{eq:71} implies that theorem \ref{bcfr} is not improvable from a theoretical point of view, in the sense that our fractal dimension II definition does not take into account the different size of the sets which appears on each level $\Gamma_{n}$ of the fractal structure $\ef$ under consideration.

\section{An application to the domain of words}\label{sec:6}
The goal of this section consists of providing a variety of applications of the fractal dimension introduced in this paper, where the box-counting dimension cannot be computed. Indeed, we show that the use of fractal structures in the context of a domain of words could contribute some information about the complexity of a language. In this way, we study and show how to calculate the fractal dimension of a language generated by means of a regular expression, and then we introduce an empirical application consisting of computing the fractal dimension of a natural language. In both cases, we show and explain the obtained results. Finally, we present how fractal structures and dimensions can help us in order to determine the efficiency of a system of information encoding like the BCD.
We start with some preliminary topics.
\subsection{Fractal structures and domains of words}
The domain of words, which we introduce next, appears when modeling the streams
of information in Kahn's model of parallel computation (see \cite{KAH74},
\cite{MAT94}).\newline
Indeed, let $\Sigma$ be a finite non-empty alphabet (set) and let $\Sigma^{\infty}$ be
the collection of finite ($\cup_{n \in \mathbb{N}} \Sigma^n$) and infinite
($\Sigma^{\mathbb{N}}$) sequences (called words) over $\Sigma$. Let us denote by
$\varepsilon$ to the empty word.\newline
The prefix order $\sqsubseteq$ is defined on $\Sigma^{\infty}$, as usual, by $x
\sqsubseteq y$ iff $x$ is a prefix of $y$. Thus, for each $x \in \Sigma^{\infty}$,
let $l(x)$ be the length of $x$, where $l(\varepsilon)=0$, and for $x,y \in
\Sigma^{\infty}$, we denote by $x \sqcap y$ to the common prefix of $x$ and $y$. Hence,
a (non-archimedean) quasi-metric $d$ can be defined on $\Sigma^{\infty}$ by
$d(x,y)=0$ if $x \sqsubseteq y$, and $d(x,y)=2^{-l(x \sqcap y)}$, in other case
(see \cite{SMI88}).\newline
Furthermore, the non-archimedean quasi-metric $d$ induces a fractal structure on
$\Sigma^{\infty}$ (see \cite{SG99}) which can be described as follows:
$\ef=\{\Gamma_n:n
\in \mathbb{N} \}$, where its levels are given by
\begin{equation}\label{eq:82}
\Gamma_n=\{w^{\#}:w \in \Sigma^n\} \cup \{w^{\sqsubseteq}:w \in \Sigma^k: k<n\}
\end{equation}
Note that for $w \in
\Sigma^n$, $w^{\sqsubseteq}=\{u \in \Sigma^k:k \leq n$ and $u \sqsubseteq w\}$ is the
collection of prefixes of $w$, and $w^{\#}=\{wu:u \in \Sigma^{\infty} \} \cup
w^{\sqsubseteq}$ is the collection of words (finite or infinite) which start
with $w$ or are a prefix of $w$.
Moreover, for each $w \in \Sigma^n$, we have that $w^{\#}=B_{d^{-1}}(w,2^{-n})$, and for each
$w \in \Sigma^k$ with $k<n$, $w^{\sqsubseteq}=B_{d^{-1}}(w,2^{-n})$.\newline
Then, a language $L$ is defined as a subspace of $\Sigma$, and usually it is defined by means of a formal grammar.
In particular, we can consider the languages generated by regular expressions.
%A simple example could be the language generates by a regular expression.\newline
Now, since we have a fractal structure and a quasi-metric, then we can calculate the
fractal dimension of any language.

%We show it in the following example.
\subsection{The fractal dimension of a language generated by a regular expression}
%\begin{ejem}\label{ejem:2}
Consider the regular expression $(00+1)^{+}$, that
is constructed by concatenating consecutively (at least one
time) $00$ and $1$. Our main purpose consists of computing the fractal dimension of the
previous language $L\subset \Sigma^{\infty}$ generated by means of the previous regular expression. In this way, we are going to apply the fractal dimension I model. First of all, note that $L$ can be described as the
following set:
$$L=\{1,00,11,100,001,111,0000,0011,1001,1100,1111,\ldots\}$$
Let also $\ef$ be the fractal structure induced by the non-archimedean quasi-metric
$d$ given at \ref{eq:82}. Then, we have that $\ef=\{\Gamma_{n}:n\in \mathbb{N}\}$,
where, for instance, the first levels are given as follows:
$\Gamma_{1}=\{1^{\#},0^{\#}\}$,
$\Gamma_{2}=\{10^{\#},11^{\#},00^{\#},1^{\sqsubseteq}\}$,
$\Gamma_{3}=\{000^{\#},001^{\#},100^{\#},111^{\#},110^{\#},00^{\sqsubseteq},\newline
11^{\sqsubseteq},1^{\sqsubseteq}\}$,
and so on.
In order to calculate the fractal dimension of $L$, note that $N_{1}(L)=2$,
$N_{2}(L)=3+1$, $N_{3}(L)=5+3$, $N_{4}(L)=8+6$, $N_{5}(L)=13+11$, and so on,
where the first term on each sum refers to the number of elements of the appropiate level of the
form $w^{\#}$, and the second one to the number of elements of each level of the
form $w^{\sqsubseteq}$. Hence, if we consider $\{f_{n}\}_{n\in
\mathbb{N}}$ as the Fibonacci's sequence, where $f_{1}=f_{2}=1$, with
$f_{n}=f_{n-1}+f_{n-2}$ for all $n\geq 3$, then it can be checked what follows:
\begin{equation}
N_{n}(L)=f_{n+2}+\sum_{i=2}^{n}f_{i}
\end{equation}
for all $n\in \mathbb{N}$. In this way, since the Fibonacci's sequence verifies the next property:
\begin{equation}
\sum_{i=1}^{n}f_{i}=f_{n+2}-1 \ \text{for all}\ n\geq 2
\end{equation}
then $N_{n}(L)=2\
(f_{n+2}-1)$ for all natural number $n$. Furthermore, it is known that
$f_{n}=\frac{\varphi^{n}-\beta^{n}}{\varphi-\beta}$ for all $n\in \mathbb{N}$,
where $\varphi=\frac{1+\sqrt{5}}{2}$ (the \emph{golden ratio}) and $\beta=\frac{1-\sqrt{5}}{2}$.
Accordingly, we obtain the fractal dimension I of $L$ as follows:
\begin{equation}\label{eq:74}
\begin{split}
\lim_{n\rightarrow \infty}\frac{\log N_{n}(L)}{n\log 2}
&=\lim_{n\rightarrow \infty}\frac{\log f_{n+2}}{n\log 2}
=\lim_{n\rightarrow \infty}\frac{\log (\varphi^{n+2}-\beta^{n+2})}{n\log 2}=\\
&=\lim_{n\rightarrow \infty}\frac{\log \Big(1-\Big(\frac{\beta}{\varphi}\Big)^{n+2}\Big)+(n+2)
\log \varphi}{n\log 2}=\log_{2}\varphi
\end{split}
\end{equation}
Hence, we have that the fractal dimension of $L$ is related with
the golden ratio. Note that the result we have just got in expression
\ref{eq:74} implies that $N_{n}(L)\simeq \varphi^{n}$ for all natural number
$n$, which is equivalent to $N_{n+1}(L)\simeq \varphi\cdot N_{n}(L)$. Roughly
speaking, we have that for any word of length $n$, there are about $\varphi$
words of length $n+1$. On the other hand, since
$\dim_{\ef}^{1}(\Sigma^{\infty})=1$, it is clear that
$N_{n+1}(\Sigma^{\infty})\simeq 2\cdot N_{n}(\Sigma^{\infty})$ for all $n\in
\mathbb{N}$. Thus, fractal dimension I constitutes a register about the
evolution and complexity of the language $L$ with respect to the domain of words
$\Sigma^{\infty}$ where it has been defined.

Moreover, since $\delta(L,\Gamma_{n})=\frac{1}{2}$ for all $n\in
\mathbb{N}$, then we have that $\dim_{\ef}^{2}(L)=\infty$, so that the fractal
dimension II method does not provide information about the evolution of the
language $L$ with respect to the domain of words.

Fractal dimension II will be useful in order to describe the complexity of a
language when the related fractal structure $\ef$ is starbase. In this way, note
that the fractal structure induced by the (non-archimedean) quasi-metric $d$ is
not starbase. Indeed, suppose the opposite. Then, $\Sigma^{\infty}$ would be a
metrizable space (see \cite{SG02B}), which implies that $\Sigma^{\infty}$ is
$T_{1}$, that constitutes a contradiction. Now, by means of proposition \ref{prop:7},
if the distance function is compatible with the fractal structure, we have that
$\delta(\Gamma_{n})\nrightarrow 0$ as $n\rightarrow \infty$. Further, if
$N_{n}(L)\rightarrow \infty$ as $n\rightarrow \infty$ (as it occurs in this
example), we get that $\dim_{\ef}^{2}(L)=\infty$.
However, this disadvantage can be improved by means of the next remark.

\begin{obs}
Let $\ef^{'}=\{\Gamma_{n}^{'}:n\in \mathbb{N}\}$ be the fractal structure whose levels are given
by $\Gamma_{n}^{'}=\{w^{\#}:w\in \Sigma^{n}\}$ for all natural number $n$. Note that this fractal structure constitutes a
simplification of the previous $\ef$. Furthermore, we have that $N_{n}^{'}(L)=f_{n+2}$, as
well as $\delta(L,\Gamma_{n})=\frac{1}{2^{n}}$ for all $n\in \mathbb{N}$, which
implies that $\dim_{\ef^{'}}^{2}(L)=\log_{2}\varphi=\dim_{\ef}^{1}(L)$, so that
fractal dimension II provides the same information about $L$ than fractal dimension I in this case.
\end{obs}
%Another interesting application of our fractal dimension models is described as follows:

\subsection{An empirical application to natural languages}
Although fractal dimensions can be computed for languages described by means
of regular expressions, where words of infinite length could exist, it is also
possible to determine the fractal dimension of other languages, like
the natural ones. In this way, since fractal dimensions are not
only theoretical but also empirical quantities, we are going to calculate them
with respect to the fractal structure $\ef$ given at \ref{eq:82}, though it is also possible to work with the fractal structure $\ef^{'}$.
Nevertheless, since always exists a maximum value for the length of any
word in a natural language, it is clear that
$\dim_{\ef}^{1}(N)=\dim_{\ef^{'}}^{1}(N)=0$ where
$N$ denotes to any natural language.

Taking it into account, for a practical application, we are going to calculate the fractal dimension of a given natural language $N$ by means of
the slope of the regression line obtained by comparing $\log N_{n}(N)$ versus $\log 2^{n}$. In order to do this, we are going to count the number of words with a given length which appears on the relative dictionary to the natural language $N$. Note that in this case we are computing $N_{n}(N)$ for $n\in \{1,2,3,4\}$, which leads to a suitable approximation of the fractal dimension of $N$. This quantity provides information about
the complexity and evolution of that natural language, by
means of the length and the variety its words present (in this case, specially
the words with small length).
%In order to calculate the fractal dimension in this case, we calculate the
%linear regression of $\log N_n(L)$ against $n \log 2$.
Indeed, we have calculated
the fractal dimension of a wide list of natural languages and got the results which appear in the table \ref{table:1}.
%(the value from fractal dimension is the number which appears in brackets):
%Polish, German (3.1), U.K. English (3.0), U.S. English, Italian, Hungarian,
%Latin (2.9), Dutch, Danish, Africans, Czech (2.8), Swedish, Spanish, Ukrainian
%(2.7), Russian, Croatian, Slovenian, Portuguese, Mongolian, Romanian (2.6),
%Serbian, French, Slovak (2.5) and Bulgarian (2.4).

\begin{table}[h]
\begin{tabular}{||l | c||}
\hline
\hline
\emph{Natural Language} & \textbf{Fractal dimension} \\
\hline
German, Polish & 3.1\\
\hline
U.K. English & 3.0\\
\hline
U.S. English, Hungarian, Italian, Latin & 2.9\\
\hline
Africans, Czech, Danish, Dutch & 2.8\\
\hline
Spanish, Swedish, Ukrainian & 2.7\\
\hline
Croatian, Mongolian, Portuguese, & \\ Romanian, Russian, Slovenian & 2.6\\
\hline
French, Serbian, Slovak & 2.5\\
\hline
Bulgarian & 2.4\\
\hline
\hline
\end{tabular}
\caption{Fractal dimensions for natural languages, calculated respect to the fractal structure $\ef$ induced by the quasi-metric $d$ (see \ref{eq:82}).}\label{table:1}
\end{table}

The obtained results allows to compare the complexity between two given natural languages. For instance,
in figure \ref{fig:1} we can regard the difference of the fractal complexity
between English and Spanish. In this way, note that English has a larger variety of
small prefixes than Spanish, although the number of prefixes of a given
length increases faster in Spanish than in English for medium length words.

\begin{center}
\begin{figure*}[here]
\begin{tabular}{c}
\includegraphics[width=115mm, height=80mm]{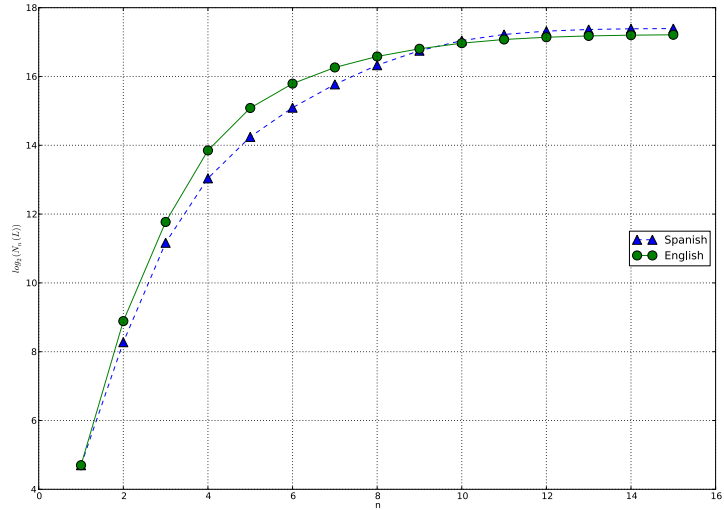}
%\\ \includegraphics[scale=0.4]{figures/graficasROC3} &
%\includegraphics[scale=0.4]{figures/graficasROC4}
\end{tabular}
\caption{Graph of the fractal complexity comparison between the languages English and Spanish.}\label{fig:1}
\end{figure*}
\end{center}

The fractal dimension models developed in this paper
allows also to compare the complexity of any text written in a specific language.
Indeed, since the fractal dimension of a language has been calculated by means
of its corresponding dictionary, an effective calculation of the fractal
dimension of a book written on that language gives us some interesting information
about the variety of the words it uses. Moreover, it is also possible to perform a fractal complexity comparison between a given translation of a book and its original version.

With this in mind, next we show some examples where the fractal dimension of some books has been calculated. In order to do this, we have considered the set which contains all the words used in the corresponding book and chosen the fractal structure $\ef$ defined in \ref{eq:82}. The effective calculation of the fractal dimension has been done by taking into account the four first levels of the fractal structure.
Thus, we show the fractal dimension of the selected books as well as the fractal dimension of that book's language (the latter is the number which appears into brackets), which gives information about the complexity of the given translations.

For instance, the fractal dimension of the English version of the
book \emph{On the Origin of Species} (Charles Darwin, $1859$) is $2.2\ (3.0)$,
while the fractal dimension of its French version is about $2.3\ (2.6)$. On the
other hand, the fractal dimension of the English version of the text
\emph{Alice's Adventures in Wonderland} (Lewis Carroll, $1865$) is $1.8\ (3.0)$,
while its Spanish version has the value $2.0\ (2.7)$. Note that a classical
book like \emph{Don Quijote de la Mancha} (Miguel de Cervantes, $1605$) has a
fractal dimension of $2.5\ (2.7)$ in its original Spanish version, while its
English translation is about $2.3\ (3.0)$, and for instance, the dimension of
its Dutch version is $2.2\ (2.9)$.

\subsection{The fractal dimension as a tool to study the efficiency of an encoding system}
So far, we have shown two applications of our models for determining the fractal
dimension where box-counting dimension has no sense. The first one results
interesting since it consists on a language generated by means of a
regular expression where infinite length words could exist. The
second one consists of an empirical application where we use the fractal dimension as a
tool in order to study the complexity and evolution of a natural language. Now, we are
going to show a computational application of our models to an interesting method
for coding and decoding information from other languages.

The binary-coded decimal (which we denote by BCD for short) is a method for encoding
decimal numbers which represents each decimal digit by means of its binary
sequence. Although nowadays is used with less frequency in some applications, it
results useful in computer and electronic systems (which specially consists only
of digital logic without microprocessors) in order to display or print decimal
numbers. One of its advantages is that BCD allows faster decimal calculations.
Nevertheless, it is not an efficient encoding method, since it uses more
space than a simple binary representation.

The BCD system stores each decimal digit from $0$ to $9$ by means of $4$ bits
which contains its binary sequence. In this way, the set
\begin{equation}
S=\{0000, 0001, 0010,
0011,0100, 0101, 0110, 0111, 1000, 1001\}
\end{equation}
generates the language $B$ of
\emph{binary-coded decimal numbers}. For instance, note that $1100\notin B$, since it
is not a codification of any decimal digit.
For example, the number $215$ (in $10$-base), is converted to an expression like
$0010\ 0001\ 0101$. Indeed, note that its binary representation is $11010111$,
which needs less space to be stored than by means of the BCD codification.

Fractal dimension constitutes a measure of the computational efficiency of a language which consists of the encoding of another language, like the BCD system. In this way, next we extract some conclusions related to the encoding goodness of the latter method. Let $\ef$ be the fractal structure defined at \ref{eq:82}, though it is also possible to consider its simplified version $\ef^{'}$. First of all, note that
\begin{equation}
N_{n}(L)=\left\{
\begin{array}{ll}
\hbox{$2\cdot 10^{i}$}\ \ldots & \hbox{$n=4i$} \\
\hbox{$3\cdot 10^{i}$}\ \ldots & \hbox{$n=4i+1$}\\
\hbox{$4\cdot 10^{i}$}\ \ldots & \hbox{$n=4i+2$}\\
\hbox{$6\cdot 10^{i}$}\ \ldots & \hbox{$n=4i+3$}\\
\end{array}
\right.
\end{equation}
which leads to $\dim_{\ef}^{1}(L)=\dim_{\ef^{'}}^{1}(L)=\frac{\log 10}{\log
16}=\log_{2}\sqrt[4]{10}$. Therefore, we have that $N_{n}(L)\simeq
10^{\frac{n}{4}}$ for all $n\in \mathbb{N}$, which implies that
$N_{n+1}(L)\simeq \sqrt[4]{10}\cdot N_{n}(L)$ for all natural number $n$. Thus,
given a real number with $n$ decimal digits, we need $4n$ binary digits in order
to encode it, which allows to represent $10^n$ different numbers on the
BCD system, while in the binary representation, $4n$ binary digits allows to
represent $2^{4n}$ numbers. Therefore, we can calculate the efficiency of BCD
when encoding decimal numbers, by means of the next ratio:
\begin{equation}
\frac{N_{n}(L)}{N_{n}(\Sigma^{\infty})}=\frac{10^{\frac{n}{4}}}{2^{n}}=\Bigg(\frac{
10}{16}\Bigg)^{\frac{n}{4}}
\end{equation}
For instance, take a decimal number with $n=10$
digits, which needs $4n=40$ bits to be encoded. Accordingly,
\begin{equation}
\frac{N_{40}(L)}{N_{40}(\Sigma^{\infty})}=\Bigg(\frac{10}{16}\Bigg)^{10}\simeq
\frac{10^{10}}{10^{12}}=0.01
\end{equation}
so that, BCD encodes the $1 \%$ of the possible
numbers encoded directly in binary. Indeed, there is a lack of the order of $62.5 \%$
(that is, $(\frac{10}{16})^{\frac{1}{4}}$) for each encoded digit.

% ----------------------------------------------------------------

\end{document}